\documentclass[a4paper,UKenglish,cleveref, autoref, thm-restate]{lipics-v2021}
\pdfoutput=1 %uncomment to ensure pdflatex processing (mandatory e.g. to submit to arXiv)
\hideLIPIcs  %uncomment to remove references to LIPIcs series (logo, DOI, ...), e.g. when preparing a pre-final version to be uploaded to arXiv or another public repository

\title{Persistent Amortised Analysis, Operationally}

\author{Anton Lorenzen}{University of Edinburgh, United Kingdom}{anton.lorenzen@ed.ac.uk}{https://orcid.org/0000-0003-3538-9688}{}%mandatory, please use full name; only 1 author per \author macro; first two parameters are mandatory, other parameters can be empty. Please provide at least the name of the affiliation and the country. The full address is optional. Use additional curly braces to indicate the correct name splitting when the last name consists of multiple name parts.

\authorrunning{A. Lorenzen} %mandatory. First: Use abbreviated first/middle names. Second (only in severe cases): Use first author plus 'et al.'

\Copyright{Anton Lorenzen} %mandatory, please use full first names. LIPIcs license is "CC-BY";  http://creativecommons.org/licenses/by/3.0/


 % mandatory; please add comma-separated list of keywords

\category{Track B: Automata, Logic, Semantics, and Theory of Programming}

\relatedversion{} %optional, e.g. full version hosted on arXiv, HAL, or other repository/website
\acknowledgements{I thank Kengo Hirata, Francois Pottier, Wouter Swierstra, and Kim Worrall for their helpful feedback on an earlier version of this paper. I thank the participants of the AUTOSARD workshop 2025 for lively discussions.}%optional

\nolinenumbers %uncomment to disable line numbering

\EventEditors{Sayan Bhattacharya, Danupon Nanongkai, Michael Benedikt, and Gabriele Puppis}
\EventNoEds{4}
\EventLongTitle{53rd International Colloquium on Automata, Languages, and Programming (ICALP 2026)}
\EventShortTitle{ICALP 2026}
\EventAcronym{ICALP}
\EventYear{2026}
\EventDate{July 7--10, 2026}
\EventLocation{Royal Holloway, University of London, Egham, United Kingdom}
\EventLogo{}
\SeriesVolume{374}
\ArticleNo{162}
\usepackage{marvosym} % for \Lightning

\usepackage{bussproofs}
\usepackage{hyperref}
\usepackage{tikz}
\usepackage{tikz-cd}
\usetikzlibrary{positioning}

\newcommand{\bitem}[1]{\item {\textbf{#1}:}}
\newcommand{\llet}{\mathop{\mathrm{let}}}
\newcommand{\lin}{\mathop{\mathrm{in}}}
\newcommand{\case}{\mathop{\mathrm{case}}}
\newcommand{\inl}{\mathop{\mathrm{inl}}}
\newcommand{\inr}{\mathop{\mathrm{inr}}}
\newcommand{\fold}{\mathop{\mathrm{fold}}}
\newcommand{\unfold}{\mathop{\mathrm{unfold}}}
\newcommand{\force}{\mathop{\mathrm{force}}}
\newcommand{\lazy}{\mathrm{lazy}}
\newcommand{\memo}{\mathrm{memo}}
\newcommand{\strip}[1]{\mathopen{\raisebox{0.4ex}{$\ulcorner$}}\mkern-4mu#1\mkern-4mu\mathclose{\raisebox{0.4ex}{$\urcorner$}}}
\newcommand{\stripdot}{\strip{\mkern+4mu\cdot\mkern+4mu}}
\newcommand{\save}[2]{\mathrm{save}\,#1\,#2}
\newcommand{\spend}[3]{\mathrm{spend}\,#1\,\mathrm{from}\,#2\,\mathrm{on}\,#3}
\newcommand{\pass}[1]{\mathrm{pass}\; #1}

\begin{document}

\maketitle

\begin{abstract}
Amortised analysis is a technique for proving a combined time bound for a batch of operations on a data structure,
even if some of those operations are expensive.
But the traditional method of amortised analysis yields incorrect time bounds when the data structure is used persistently.
Persistence allows operations to be performed on previous versions of the data structure,
which prevents us from amortising expensive restructuring work.
In his seminal book, Chris Okasaki showed how to extend amortised analysis to persistent usage.
His method works by extending the data structure with \textit{thunks}
and performing the analysis with \textit{debits} rather than credits.
His argument, that credits are unsound for analysing persistent usage, has become folklore.

In this paper, we provide a new perspective on the role of debits in Okasaki's work.
First, we set up an operational semantics of call-by-value lambda calculus with thunks,
and show formally that traditional amortised analysis does not work in a persistent setting.
Then we show that, contrary to the folklore, credit-based amortised analysis
can be sound in a persistent setting as long as credits are only stored on thunks.
Finally, we provide a formal semantics for Okasaki's debit-based approach.
Our paper clarifies the formal foundation of Okasaki's work
and makes it accessible to a wider audience.
\end{abstract}

\section{Introduction}
\label{sec:intro}

Amortised analysis is a technique for proving a combined time bound for a batch of operations on a data structure,
even if some of those operations are expensive.
But the traditional method of amortised analysis only yields correct time bounds
if data structures are used \textit{sequentially} and fails if data structures are used \textit{persistently}.
A data structure is said to be used sequentially if each operation is applied only to the most recent version of the data structure,
as returned by the operation preceding it. This happens automatically in mutable data structures,
where operations overwrite the data structure and previous versions are no longer accessible.

In contrast, persistent data structures~\cite{Driscoll:persistent,Sarnak:persistent,Driscoll:catenation}
enable the programmer to access or update any previous version of the data structure at no additional cost.
This creates a more flexible programming model, but breaks amortised analysis.
Consider a version of a persistent data structure that is highly unbalanced.
In a traditional amortised analysis, this version is assigned \textit{credits},
which may be spent by a future operation to perform expensive restructuring work
and rebalance the data structure~\cite{Tarjan:amortisation}.
But credits may be spent on this work only if the previous version is not used again.
If the previous version is still accessible, each future operation on it
will have to perform the same expensive restructuring work again.
Since the stored credits can only pay for one of these operations,
traditional amortised analysis yields incorrect time bounds for data structures that are used persistently.

Okasaki~\cite{Okasaki:purefun} proposed that \textit{thunks} can be used to achieve amortised time bounds for data structures that are used persistently.
A thunk is a memory cell that may only be updated in a special, deterministic fashion:
It starts out in a \textit{lazy} state and may be updated to the \textit{memoised} state.
Its value in the memoised state must be a deterministic function of its value in the lazy state,
and it should be unobservable whether the update has already happened;
a program may ask to read the memoised value, but it will not learn whether an update was performed to obtain it.
This makes it possible to share restructuring work between different versions of a data structure.
If an operation memoises a thunk, it will be memoised in all previous versions of the data structure as well.
This does not change the behaviour of the previous versions, since the update is deterministic and unobservable.

In his seminal book~\cite{Okasaki:purefun}, Okasaki adds thunks to many well-known data structures
and shows that his variations enjoy good amortised time bounds, even when used persistently.
In his analyses, he reasons as if thunks were updated to the memoised state immediately after they are created.
However, he allows the program to access the memoised value only once the update has been paid for.
The cost of the update is recorded as \textit{debits} on the thunk.
A debit is a negative credit that costs one time step to reduce by one,
and the thunk can only be accessed once the debits have been reduced to zero.

This technique allows Okasaki to reason about persistent usage as if the data structure was only used sequentially.
That is, if a data structure observes a particular time bound in a sequential setting according to his reasoning style,
then it also observes the same time bound in a persistent setting.
He claims that this is due to his use of debits rather than credits:
``although savings can only be spent once, it does no harm to pay off debt more than once''~\cite[page~59]{Okasaki:purefun}.
This claim has become folklore and led debits to be the foundation of subsequent work~\cite{Danielsson:amortisation,Pottier:thunks,Jost:lazy,Mevel:timecredits,Pilkiewicz:monotonic}.

Okasaki only provides an informal argument for the correctness of his reasoning style,
which has raised an interest in formalising his approach.
Danielsson~\cite{Danielsson:amortisation} implements a variant of Okasaki's debit passing style in Agda.
He shows that it is \textit{sound}: the time bound proven using debit passing style
is an upper bound on the actual time taken by the program.
Pottier et al.~\cite{Pottier:thunks} formalise this variant of debit passing style in Iris
and additionally provide a formal foundation for reasoning about the future evaluation of thunks.
However, these approaches differ from (and improve upon) Okasaki's approach:
they assume that thunks are updated only when the program first requests their value.

With this improvement, it is clearer why one may reason about persistent usage as if the data structure was only used sequentially.
A thunk is only updated once, which makes it possible to store credits on thunks
and spend them on the update once it is requested~\cite{Danielsson:amortisation,Pilkiewicz:monotonic,Pottier:thunks,Mevel:timecredits}.
Pilkiewicz and Pottier~\cite{Pilkiewicz:monotonic} and M{\'e}vel et al.~\cite{Mevel:timecredits} prove that thunks are \textit{monotonic}: 
When reasoning about a persistent data structure, we may assume that a thunk is in the lazy state and requires further credits before it can be updated,
but if the thunk is already memoised, our reasoning will still be valid and simply over-approximate the cost.
This insight makes it possible to reason about persistent usage using credits rather than debits.
Nevertheless, their high-level reasoning principles are still based on debits.

In this paper, we provide a new perspective on the role of debits in Okasaki's work.

In \cref{sec:persistence}, we define the soundness and persistence of a reasoning principle
in terms of the operational semantics of call-by-value lambda calculus with thunks.
By deriving all of our results directly from an operational semantics,
we can state exactly when a reasoning principle works in a persistent setting.
As an example, we re-derive the fact that the traditional method of amortised analysis is not persistent.

In \cref{sec:credits}, we give operational semantics for credit passing style and credit inheritance.
Credit passing style is a variant of Danielsson's debit passing style, which places \textit{credits on lazy thunks}.
We show that it is sound and persistent in our operational model.
Contrary to the folklore, this shows that persistent amortised analysis can be performed purely in terms of credits.

In \cref{sec:debits}, we provide an operational semantics of Okasaki's persistent Bankers method.
Rather than store credits on lazy thunks, this method stores \textit{debits on memoised thunks}.
This models Okasaki's use of debits as a barrier or ``layaway plan'',
which prevents the program from accessing the thunk before it has been paid off.
We show that this is sound and persistent in our operational semantics.

Our earlier workshop paper~\cite{Lorenzen:creditmonad} contains an implementation of the
credit-based reasoning principles as a Haskell library for automated testing of time complexity bounds.
As part of that work, we have re-analysed the time complexity of all lazy data structures in Okasaki's book~\cite{Okasaki:purefun} using credits.
Our library is available online via the \href{https://hackage.haskell.org/package/creditmonad}{creditmonad} package
and on \href{https://github.com/anfelor/creditmonad}{GitHub}.

\section{Preliminaries}
\label{sec:prelim}

To model operations on persistent data structures, we consider a lambda calculus
with algebraic data types.
This is convenient, since the lambda calculus does not include any operations
that mutate data structures and so all data structures are persistent by default.

Our syntax is split into values, heap values and expressions.
Expressions denote computations, which can return a value and store heap values in the heap.

\begin{tabular}{l r l r l r l r}
$v,w$ & ::=    & $x,y,z$          & (variables)                &\quad $e$ & ::=      & $v$                             & (values)         \\
    & $\mid$ & $a,b,c$          & (pointers)                   &        &  $\mid$  & $hv$                            & (heap values) \\
    & $\mid$ & $()$             & (unit)                       &        &  $\mid$  & $\llet x = e \lin e$            & (let binding)           \\
    & $\mid$ & \multicolumn{2}{l}{$\inl v \mid \inr v\;\qquad$ (sum)} & &  $\mid$  & \multicolumn{2}{l}{$\case v \,\{ \inl x \mapsto e; \inr y \mapsto e \}\;\;\;\;\;\;$ (case split)}       \\
    & $\mid$ & $(v, v)$         & (pair)                       &        &  $\mid$  & $\llet (x, y) = v \lin e $           & (destruct a pair)       \\
    & $\mid$ & $\lambda x. e$   & (lambda)                     &        &  $\mid$  & $v\,w$                            & (lambda application)      \\
$hv$ & ::=   & $\fold v$        & (allocate)    &        &  $\mid$  & $\unfold v$                            & (dereference) \\       
$\mathcal{F}$ & ::= & \multicolumn{2}{l}{$\emptyset \mid \mathcal{F}, F(x) = e$}       &        &  $\mid$  & $F\,v$                                  & (function call)           \\
\end{tabular}

We distinguish between variables $x,y,z$ in computations and pointers $a,b,c$ in the heap.
Loosely speaking, sums allow us to choose between values (similar to enumerations or tagged unions),
while pairs allow us to combine values (similar to structs).
We will not need lambdas (which correspond to anonymous functions) in this article, but include them for completeness.
We define top-level functions in the $\mathcal{F}$ environment and call them using $F(v)$.
We do not include a fixpoint operator or other recursion scheme, and instead assume that top-level functions may be recursive.
We use $\fold$ to allocate an expression into the heap and $\unfold$ to look up a value from the heap.

Our syntax uses fine-grain call-by-value style~\cite{Levy:fgcbv}, where most primitives operate on values.
This makes it easier to specify semantic rules. However, one can easily desugar a more traditional syntax
into this style by introducing let-bindings where necessary.

\subsection{Big-Step Semantics}

The (worst-case) time complexity of our language is given by a big-step operational semantics.
We write $\Gamma : e \Downarrow_k \Delta : w$ to mean ``under heap $\Gamma$,
the expression $e$ evaluates to value $w$ with updated heap $\Delta$ in $k$ steps''.
A heap is defined as a mapping from pointers to heap values:
\begin{equation*}
\Gamma ::= \emptyset \mid \Gamma, a \mapsto hv
\end{equation*}

Our big-step semantics is specified using the inference rules below.
Each rule specifies that if the premises above the line hold,
then the conclusion below the line also holds.
We write $[v/x]$ for the capture-avoiding substitution that replaces all free occurrences of $x$ with $v$.
We may assume that each operation is performed on values of the correct shape:
for example, a case-split is performed on either an $\inl v$ or an $\inr v$ value.
If the shape of the value is incorrect, no rule applies and we say that the evaluation is stuck.
We only write $\Gamma : e \Downarrow_k \Delta : w$ if the evaluation is not stuck.

\begin{center}
  \AxiomC{\phantom{$\Downarrow$}}
  \RightLabel{(value)}
  \UnaryInfC{$\Gamma : v \Downarrow_0 \Gamma : v$}
  \DisplayProof
  \quad
  \AxiomC{$\Gamma : e_i[v/x_i] \Downarrow_k \Delta : w$}
  \RightLabel{(case)}
  \UnaryInfC{$\Gamma : \case (\lin_i v) \,\{ \lin_l x_l \mapsto e_l; \lin_r x_r \mapsto e_r \} \Downarrow_{k} \Delta : w$}
  \DisplayProof
\end{center}

\begin{center} 
  \AxiomC{$a \notin \mathrm{dom}(\Gamma)$}
  \RightLabel{(fold)}
  \UnaryInfC{$\Gamma : \fold v \Downarrow_0 \Gamma, a \mapsto \fold v : a$}
  \DisplayProof
  \quad
  \AxiomC{$\Gamma : e_1 \Downarrow_k \Delta : v$}
  \AxiomC{$\Delta : e_2[v/x] \Downarrow_l \Theta : w$}
  \RightLabel{(let)}
  \BinaryInfC{$\Gamma : \llet x = e_1 \lin e_2 \Downarrow_{k+l} \Theta : w$}
  \DisplayProof
\end{center}

\begin{center}
  \AxiomC{$a \mapsto \fold v \in \Gamma$}
  \RightLabel{(unfold)}
  \UnaryInfC{$\Gamma : \unfold a \Downarrow_0 \Gamma : v$}
  \DisplayProof
  \quad
  \AxiomC{$\Gamma : e[v_1/x, v_2/y] \Downarrow_k \Delta : w$}
  \RightLabel{(split)}
  \UnaryInfC{$\Gamma : \llet (x, y) = (v_1, v_2) \lin e \Downarrow_{k} \Delta : w$}
  \DisplayProof
\end{center}

\begin{center} 
  \AxiomC{$\Gamma : e[v/x] \Downarrow_k \Delta : w$}
  \RightLabel{(app)}
  \UnaryInfC{$\Gamma : (\lambda x. e) v \Downarrow_{k+1} \Delta : w$}
  \DisplayProof
  \qquad
  \AxiomC{$F(x) = e \in \mathcal{F}$}
  \AxiomC{$\Gamma : e[v/x] \Downarrow_k \Delta : w$}
  \RightLabel{(call)}
  \BinaryInfC{$\Gamma : F\,v \Downarrow_{k+1} \Delta : w$}
  \DisplayProof
\end{center}

In rules that allocate values into the heap, we require that the allocated pointer $a$ is fresh.
To simplify matters, it is convenient to assume that this choice is deterministic.
When we compare two evaluations of the same expression, we may assume that they allocate their heap values into the same memory cells.
In addition, we assume that the chosen pointer is globally fresh.
For example, when we have two evaluations of expressions $e_1$ and $e_2$ from the same heap,
we may combine them into an evaluation of $\llet x = e_1 \lin e_2$
by assuming that they never choose the same pointer for allocation.
It is possible (but laborious) to lift both assumptions through the explicit maintenance of substitutions of pointers.

This operational semantics closely approximates the cost of executing a functional program in practice.
Like Danielsson~\cite{Danielsson:amortisation}, we only count the steps that correspond to lambda application and function calls.
It can be shown that, for a fixed expression, call-by-value semantics can be simulated by a random access machine
using only constant overhead~\cite{Blelloch:parallelism}.

\subsection{Thunks}

To model thunks, we extend our syntax and semantics. We add two new heap values: lazy thunks and memoised thunks.
A lazy thunk holds a value and is annotated by the top-level function $F$ that should be applied to update the thunk.
A memoised thunk holds the value that was computed by the update.
We can force a thunk to obtain its memoised value, updating it if necessary.
In practice, the function $F$ is not stored in the thunk itself,
but inferred from its type during the force operation~\cite{Lorenzen:folazy}.

\begin{tabular}{l r l l l r l r}
$hv$ & ::= $\ldots \mid$ & $\memo\; v$   & (memoised value)             &\quad $e$  & ::= $\ldots \mid$ & $\force v$                          & (force thunk)         \\
     & $\mid$ & $\lazy_F\; v$            & (lazy computation)           &           &                                                       & \\
\end{tabular}

To extend our semantics with thunks, we use a variant of Launchbury's natural semantics~\cite{Launchbury:lazysem},
which makes thunks first-order~\cite{Lorenzen:folazy}.
We allocate a thunk in a fresh memory cell using the (lazy) and (memo) rules.
If a thunk is memoised, we retrieve its value using the (recall) rule at a cost of one step.
If a thunk is lazy, we have to force it using the (force) rule,
where we remove it from the heap, run the computation and then store the memoised value back into the heap.
Forcing itself is free, but involves a function call, which costs one step.

\begin{center}
    \AxiomC{$F \in \mathcal{F}$}
    \AxiomC{$a \not\in \mathrm{dom}(\Gamma)$}
    \RightLabel{(lazy)}
    \BinaryInfC{$\Gamma : \lazy_F\; v \Downarrow_0 (\Gamma, a \mapsto \lazy_F\; v) : a$}
    \DisplayProof
\end{center}
\begin{center}
    \AxiomC{$\Gamma : F\, v \Downarrow_k \Delta : w$}
    \RightLabel{(force)}
    \UnaryInfC{$(\Gamma, a \mapsto \lazy_F\; v) : \mathrm{force}\; a \Downarrow_{k} (\Delta, a \mapsto \memo\; w) : w$}
    \DisplayProof
\end{center}
\begin{center}
    \AxiomC{$a \not\in \mathrm{dom}(\Gamma)$}
    \RightLabel{(memo)}
    \UnaryInfC{$\Gamma : \memo\; v \Downarrow_0 (\Gamma, a \mapsto \memo\; v) : a$}
    \DisplayProof
    \qquad
    \AxiomC{$a \mapsto \memo\; v \in \Gamma$}
    \RightLabel{(recall)}
    \UnaryInfC{$\Gamma : \mathrm{force}\; a \Downarrow_0 \Gamma : v$}
    \DisplayProof
\end{center}

\begin{example}
Because the (force) rule changes thunks from lazy computations to memoised values,
repeatedly forcing a thunk is cheap. For example, we can combine several rules to derive:

\begin{center}
  % \RightLabel{(lazy)}
  % \UnaryInfC{$\Gamma : \lazy_F\; v \Downarrow_0 (\Gamma, a \mapsto \lazy_F\; v) : a$}
  \AxiomC{$\Gamma : F\, v \Downarrow_k \Delta : w$}
  \AxiomC{$a \notin \mathrm{dom}(\Gamma)$}
  % \RightLabel{(force)}
  % \UnaryInfC{$(\Gamma, a \mapsto \lazy_F\; v) : \mathrm{force}\; a \Downarrow_{k} (\Delta, a \mapsto \memo\; w) : w$}
  % \AxiomC{$a \mapsto \memo\; w \in (\Delta, a \mapsto \memo\; w)$}
  % \RightLabel{(recall)}
  % \UnaryInfC{$(\Delta, a \mapsto \memo\; w) : \mathrm{force}\; a \Downarrow_0 (\Delta, a \mapsto \memo\; w) : w$}
  % \RightLabel{(let)}
  % \BinaryInfC{$(\Gamma, a \mapsto \lazy_F\; v) : \llet y = \mathrm{force}\; a \lin \mathrm{force}\; a \Downarrow_{k} (\Delta, a \mapsto \memo\; w) : w$}
  \RightLabel{(...)}
  \BinaryInfC{$\Gamma : \llet x = \lazy_F\; v \lin \llet y = \mathrm{force}\; x \lin \mathrm{force}\; x \Downarrow_{k} (\Delta, a \mapsto \memo\; w) : w$}
  \DisplayProof
\end{center}
Despite calling $\mathrm{force}\; x$ twice, the total cost is only $k$ steps instead of $2k$ steps.
\end{example}

\section{Persistent Amortised Analysis}
\label{sec:persistence}

To illustrate the traditional method of amortised analysis and why it fails in a persistent setting,
consider a binary counter implemented as a (little-endian) list of bits.
The increment operation traverses the list until it finds a zero bit
and flips all bits that it encounters along the way.

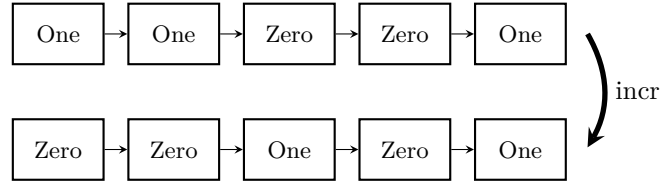
\vspace{-0.2cm}
\begin{center}
\begin{tikzpicture}[
  node distance=0.7cm and 1cm,
  counter node/.style={rectangle, draw=black, thick, minimum width=1.2cm, minimum height=0.8cm, font=\small},
  label node/.style={font=\small},
  >=stealth,
]

  % Left binary counter (11001) - vertical
  \node[counter node] (c1_1) {One};
  \node[counter node, right=0.3cm of c1_1] (c1_2) {One};
  \node[counter node, right=0.3cm of c1_2] (c1_3) {Zero};
  \node[counter node, right=0.3cm of c1_3] (c1_4) {Zero};
  \node[counter node, right=0.3cm of c1_4] (c1_5) {One};
  
  % Links between left counter nodes
  \draw[->] (c1_1) -- (c1_2);
  \draw[->] (c1_2) -- (c1_3);
  \draw[->] (c1_3) -- (c1_4);
  \draw[->] (c1_4) -- (c1_5);

  % Transformation arrow
  \draw[->, line width=2pt, bend left] (7, 0) to node[midway, right=0cm] {incr} (7, -1.5);

  % Right binary counter (00101) - vertical
  \node[counter node, below=0.7cm of c1_1] (c2_1) {Zero};
  \node[counter node, right=0.3cm of c2_1] (c2_2) {Zero};
  \node[counter node, right=0.3cm of c2_2] (c2_3) {One};
  \node[counter node, right=0.3cm of c2_3] (c2_4) {Zero};
  \node[counter node, right=0.3cm of c2_4] (c2_5) {One};
  
  % Links between right counter nodes
  \draw[->] (c2_1) -- (c2_2);
  \draw[->] (c2_2) -- (c2_3);
  \draw[->] (c2_3) -- (c2_4);
  \draw[->] (c2_4) -- (c2_5);

\end{tikzpicture}
\end{center}

A counter representing the number $n$ has at most $\log n$ bits
and incrementing the counter takes $O(\log n)$ time in the worst case.
However, most increments will be far quicker.
In fact, incrementing a binary counter from zero to $n$ takes only $O(n)$ time in total,
when the worst-case bound would suggest that it takes $O(n \log n)$ time.
This motivates the Bankers method of amortised analysis~\cite{Tarjan:amortisation},
which shows that the \textit{amortised time} per increment is only $O(1)$.

\subsection{Bankers Method}

The Bankers method makes it possible to amortise expensive operations against cheap operations.
Cheap operations may spend additional time to save \textit{credits} on the heap.
Credits may be spent during an expensive operation to reduce its cost.
We model credits by a natural number $n$ that is attached to a heap allocation:
\begin{equation*}
\Gamma ::= \emptyset \mid \Gamma, a \mapsto_n hv
\end{equation*}

We add rules for saving credits on these cells and spending them later:

\begin{center} 
  \AxiomC{}
  \RightLabel{(save)}
  \UnaryInfC{$\Gamma, a \mapsto_n hv : \save{m}{a} \Downarrow_{m} \Gamma, a \mapsto_{n+m} hv : a$}
  \DisplayProof
\end{center}
\begin{center} 
  \AxiomC{$\Gamma, a \mapsto_{n} hv : e \Downarrow_k \Delta : w$}
  \AxiomC{$n,m \geq 0$}
  \RightLabel{(spend)}
  \BinaryInfC{$\Gamma, a \mapsto_{n+m} hv : \spend{m}{a}{e} \Downarrow_{\max(k-m,0)} \Delta : w$}
  \DisplayProof
\end{center}

In the (save) rule, we take $m$ steps to add $m$ credits to a heap allocation.
Conversely, in the (spend) rule, we may spend $m$ credits from the heap to reduce
the number of steps we have to take. We do not allow negative credits:
if we try to spend more credits than we have, this rule does not apply
and evaluation is stuck.

\subsection{Soundness}

To show that the (save) and (spend) rules are correct,
we need to compare them to the rules defined earlier.
The big-step rules defined in \cref{sec:prelim} give the worst-case number of steps
that are needed to evaluate an expression under a given heap.
We will call it the real semantics and denote it by $\Downarrow^{\mathrm{R}}_k$.
The Bankers semantics extends the real semantics with the (save) and (spend) rules,
and we denote it by $\Downarrow^{\mathrm{B}}_k$.

We show that (save) and (spend) rules are correct by comparing
the Bankers semantics to the real semantics.
But this is tricky to do directly, since they use
different syntaxes for heaps.
We relate the heaps through an erasure function
$\stripdot : \mathrm{Heap}^\mathrm{B} \to \mathrm{Heap}^\mathrm{R}$
that embeds the Bankers heaps with credits into the real heaps without credits.

\begin{definition}[Cost Model]
A cost model $(\Downarrow, \Phi, \stripdot)$ consists of:
\begin{itemize}
  \item a class of heaps $\mathrm{Heap}$ which may be embedded into real heaps by $\stripdot : \mathrm{Heap} \to \mathrm{Heap}^\mathrm{R}$
  \item a big-step operational semantics $\Downarrow : (\mathrm{Heap} \times \mathrm{Expr}) \rightharpoonup \mathbb{N} \times \mathrm{Heap} \times \mathrm{Value}$
  \item a potential function $\Phi : \mathrm{Heap} \to \mathbb{N}$.
\end{itemize}
\end{definition}

We call a cost model \textit{sound} if it yields the same result as the real semantics
and allows us to prove an upper bound on the number of steps taken by the real semantics.

\begin{definition}[Soundness~\cite{Danielsson:amortisation}]
A cost model $(\Downarrow, \Phi, \stripdot)$ is \textit{sound}
if for all $\Gamma : e \Downarrow_n \Delta : v$:
\begin{itemize}
  \item $\strip{\Gamma} : e \Downarrow^{\mathrm{R}}_k \strip{\Delta} : v$
  \item $k + \Phi(\Delta) - \Phi(\Gamma) \leq n$.
\end{itemize}
\end{definition}

\begin{corollary}
  The real cost model $(\Downarrow^{\mathrm{R}}, \Phi^{\mathrm{R}}, \stripdot^{\mathrm{R}})$ is sound
  for $\Phi^{\mathrm{R}}(\Gamma) = 0$ and $\strip{\Gamma}^{\mathrm{R}} = \Gamma$.
\end{corollary}

The Bankers method yields a sound cost model.
The potential function $\Phi^{\mathrm{B}}$ counts the total number of credits in the heap
and the $\stripdot^{\mathrm{B}}$ function removes the credit annotations:
\begin{equation*}
  \Phi^{\mathrm{B}}(\emptyset) = 0 \quad\quad \Phi^{\mathrm{B}}(\Gamma, a \mapsto_n hv) = \Phi^{\mathrm{B}}(\Gamma) + n \quad\quad
  \strip{\emptyset}^{\mathrm{B}} = \emptyset \quad\quad \strip{\Gamma, a \mapsto_n hv}^{\mathrm{B}} = \strip{\Gamma}^{\mathrm{B}}, a \mapsto hv
\end{equation*}

Additionally, we need to interpret the `$\save{m}{a}$' and `$\spend{m}{a}{e}$' expressions in the real semantics.
We treat them as no-ops, where `$\save{m}{a}$' returns $a$ at no cost
and `$\spend{m}{a}{e}$' evaluates $e$ without adjusting its cost.

\begin{lemma}[\cite{Tarjan:amortisation}]
\label{lem:bankers-sound}
The Bankers cost model $(\Downarrow^{\mathrm{B}}, \Phi^{\mathrm{B}}, \stripdot^{\mathrm{B}})$ is sound.
\end{lemma}
\begin{proof}
By induction on the derivation of $\Gamma : e \Downarrow^{\mathrm{B}}_n \Delta : v$.
The (save) and (spend) rules do not change the final result.
(save) adds $m$ credits to the heap but also costs $m$ steps.
(spend) reduces the number of steps by up to $m$ but also removes $m$ credits from the heap.
\end{proof}

\subsection{Amortised Analysis of Binary Counters}

To analyse the binary counter using the Bankers method,
we define the heaps of binary counters inductively.
The $\mathrm{Counter}(\Gamma : v)$ predicate defines a counter
as a list of zeros and ones, where we store a credit for each one-bit.
\begin{center}
$\texttt{End} = \inl ()$ \quad
$\texttt{Cons}(n, a) = \inr (n, a)$ \quad
$\texttt{Zero} = \inl ()$ \quad
$\texttt{One} = \inr ()$
\end{center}
\vspace{-0.2cm}
\begin{align*}
&\mathrm{Counter}(a \mapsto_0 \fold \texttt{End} : a) & \\
\mathrm{Counter}(\Gamma : a) \implies &\mathrm{Counter}(\Gamma, b \mapsto_0 \fold \texttt{Cons}(\texttt{Zero}, a) : b) & (b \notin \mathrm{dom}(\Gamma)) \\
\mathrm{Counter}(\Gamma : a) \implies &\mathrm{Counter}(\Gamma, b \mapsto_1 \fold \texttt{Cons}(\texttt{One}, a) : b) & (b \notin \mathrm{dom}(\Gamma))
\end{align*}

We define the increment operation as follows.
We must save a credit for each one-bit we create,
but may spend a credit for each one-bit that we flip to a zero-bit.
\begin{align*}
&\mathrm{incr}(c) = \case (\unfold c) \{ \\
&\quad \texttt{End} \mapsto \save{1}{(\fold \texttt{Cons}(\texttt{One}, c))} \\
&\quad \texttt{Cons}(n, a) \mapsto \case n \{\\
&\quad\quad \texttt{Zero} \mapsto \save{1}{(\fold \texttt{Cons}(\texttt{One}, a))} \\
&\quad\quad \texttt{One}  \mapsto \spend{1}{c}{ \fold \texttt{Cons}(\texttt{Zero}, \mathrm{incr}(a)) } \,\} \,\}
\end{align*}

\begin{lemma}
\label{lem:counter-sequentially}
If $\mathrm{Counter}(\Gamma : c)$, then $\Gamma : \mathrm{incr}(c) \Downarrow^{\mathrm{B}}_{2} \Delta : c'$
with $\mathrm{Counter}(\Delta : c')$.
\end{lemma}
\begin{proof}
By induction on the $\mathrm{Counter}$ predicate. In each case, the function call takes one step.
\begin{itemize}
  \bitem{(End)}
    From $\mathrm{Counter}(\Gamma : c)$, we obtain
    $\mathrm{Counter}(\Gamma, c' \mapsto_1 \fold \texttt{Cons}(\texttt{One}, c) : c')$.
    We save one credit on the new one-bit, for a total cost of $2$.
  \bitem{(Zero)} 
    From $\mathrm{Counter}(\Gamma, c \mapsto_0 \fold \texttt{Cons}(\texttt{Zero}, a) : c)$,
    we deduce $\mathrm{Counter}(\Gamma : a)$ and obtain
    $\mathrm{Counter}(\Gamma, c' \mapsto_1 \fold \texttt{Cons}(\texttt{One}, a) : c')$.
    We save one credit on the new one-bit, for a total cost of $2$.
  \bitem{(One)}
    From $\mathrm{Counter}(\Gamma, c \mapsto_1 \fold \texttt{Cons}(\texttt{One}, a) : c)$,
    we deduce $\mathrm{Counter}(\Gamma : a)$.
    Thus, by the induction hypothesis,
    $\Gamma : \mathrm{incr}(a) \Downarrow^{\mathrm{B}}_{2} \Delta : c'$ with $\mathrm{Counter}(\Delta : c')$.
    We obtain $\mathrm{Counter}(\Delta, c'' \mapsto_0 \fold \texttt{Cons}(\texttt{Zero}, c') : c'')$.
    This yields a total cost of $3$, but we can reduce it to $2$ by spending the credit from the one-bit.
\end{itemize}
\end{proof}

\subsection{Persistent Usage}

This proof is standard for imperative languages, where each increment operation
modifies the counter in place. However, this proof does not apply if the counter is used persistently.
In functional languages it is possible to increment the same counter multiple times:

\begin{lstlisting}
val ones = One(One(One(End))) // saves three credits
val c1 = incr(ones)           // spends three credits and saves one
val c2 = incr(ones)           // spends three credits and saves one
\end{lstlisting}

In the program above, we create a counter with three one-bits, which gives us three credits.
Then we perform two increments on \textit{the same counter}.
Each increment flips three one-bits to zero-bits and creates another one-bit.
In total, the program spends six credits but only saves five credits.
Clearly this is wrong: we cannot spend more credits than we have.

The problem is that \cref{lem:counter-sequentially} does not apply to the second increment.
Given $\mathrm{Counter}(\Gamma : \text{ones})$, we can perform the first increment,
yielding $\mathrm{Counter}(\Delta : c_1)$. But there is no guarantee that
$\mathrm{Counter}(\Delta : \text{ones})$ holds, which prevents us from applying the lemma again.

In fact, increments do not take $O(1)$ amortised time in a persistent setting:

\begin{proposition}
\label{prop:counter-not-persistent}
Starting from an empty counter, $n$ persistent increments may take $\Omega(n\log n)$ time.
\end{proposition}
\begin{proof}
Use up to $n/2$ increments to obtain a binary counter corresponding to the number $2^{\lfloor \log (n/2) \rfloor} - 1$.
Then perform $n/2$ increments on this counter, persistently.
Each increment takes $\Omega(\log n)$ time, since it has to flip all $\lfloor \log (n/2) \rfloor$ bits again.
\end{proof}

This is not just of purely theoretical concern: many classic data structures do not
enjoy good amortised time bounds in a persistent setting.
For example, Binomial Heaps~\cite{Vuillemin:binomial} extend binary counters to a
priority queue by associating each one-bit with a tree.
However, their amortised bound does not hold in a persistent setting~\cite{Okasaki:purefun}.

\subsection{Persistence}

To describe when an analysis holds in a persistent setting,
we need to formalise how the heap may change during evaluation.
Intuitively, an operational semantics is persistent if the heap only gets ``better'' during evaluation.
First, we describe how the heap changes during evaluation using a preorder $\sqsubseteq$ on heaps,
which we call the \textit{accessibility} relation~\cite{Ahmed:semantics}.

\begin{definition}[Accessibility]
An accessible cost model is a cost model $(\Downarrow, \Phi, \stripdot)$ equipped with a preorder $\sqsubseteq$ on heaps such that
$\Gamma \sqsubseteq \Delta$ iff $\Gamma : e \Downarrow_n \Delta : v$ for some $e,n,v$.
\end{definition}

Note that $\sqsubseteq$ is fully determined by $\Downarrow$.
However, it is convenient to have a succinct definition of $\sqsubseteq$
that hides the concrete expression $e$.
For this reason, we will define the $\sqsubseteq$ relation explicitly using inference rules,
and prove that our explicit rules yield accessible cost models.
For the real cost model, the heap only changes by adding new allocations and forcing thunks:

\begin{center}
  \AxiomC{\phantom{$\Gamma \sqsubseteq \Delta$}}
  \RightLabel{[refl]}
  \UnaryInfC{$\Gamma \sqsubseteq \Gamma$}
  \DisplayProof
  \qquad
  \AxiomC{$\Gamma \sqsubseteq \Delta$}
  \AxiomC{$\Delta \sqsubseteq \Theta$}
  \RightLabel{[trans]}
  \BinaryInfC{$\Gamma \sqsubseteq \Theta$}
  \DisplayProof
  \qquad
  \AxiomC{$a \notin \mathrm{dom}(\Gamma)$}
  \RightLabel{[alloc]}
  \UnaryInfC{$\Gamma \sqsubseteq \Gamma, a \mapsto hv$}
  \DisplayProof
\end{center}
\begin{center}
  \AxiomC{$\Gamma \sqsubseteq \Delta$}
  \AxiomC{$\Gamma : F\, v \Downarrow_k \Delta : w$}
  \RightLabel{[force]}
  \BinaryInfC{$(\Gamma, a \mapsto \lazy_F\; v) \sqsubseteq (\Delta, a \mapsto \memo\; w)$}
  \DisplayProof
\end{center}

Let us say that an analysis describes a sequence of operations $F_1,F_2,\ldots,F_m$,
and we ensure that each intermediate state $\Gamma_i : v_i$ satisfies some invariant.
This is enough to show soundness in a sequential setting.
In a persistent setting, an operation may change the intermediate state $\Gamma_i : v_i$ to $\Delta_i : v_i$
for $\Gamma_i \sqsubseteq \Delta_i$. Can we still apply the operation $F_i$ to $\Delta_i : v_i$?

\begin{definition}[Persistence]
An accessible cost model $((\Downarrow, \Phi, \stripdot), \sqsubseteq)$ is \textit{persistent} if
for all $\Gamma : e \Downarrow_n \Gamma' : v$ and $\Gamma \sqsubseteq \Delta$,
there exist $\Delta', k$ such that $\Delta : e \Downarrow_k \Delta' : v$, $k \leq n$ and $\Gamma' \sqsubseteq \Delta'$.
It is \textit{uniformly persistent} if $k = n$.
\end{definition}

If the cost model is persistent, we can apply $F_i$ to $\Delta_i : v_i$.
It is guaranteed that this yields the same result $v_{i+1}$, does not take more steps,
and that the new state $\Delta_{i+1} : v_{i+1}$ is still accessible from $\Gamma_{i+1} : v_{i+1}$.
This allows us to inductively shift all subsequent operations to the extended heap:
\begin{center}
\begin{tikzcd}[row sep=large, column sep=7em, every label/.append style={scale=1.3}]
\Gamma_1 \arrow[r, "F_1\, v_1 \Downarrow_{n_1} v_2"] & \Gamma_2 \arrow[r, "F_2\, v_2 \Downarrow_{n_2} v_3"] \arrow[d, "\sqsubseteq"] & \Gamma_3 \arrow[r, "F_3\, v_3 \Downarrow_{n_3} v_4"] \arrow[d, dashed, "\sqsubseteq"] & \cdots \arrow[r, "F_m\, v_m \Downarrow_{n_m} v_{m+1}"] \arrow[d, dashed, "\sqsubseteq"] & \Gamma_{m+1} \arrow[d, dashed, "\sqsubseteq"] \\
                                                     & \Delta_2 \arrow[r, dashed, "F_2\, v_2 \Downarrow_{k_2} v_3"']                 & \Delta_3 \arrow[r, dashed, "F_3\, v_3 \Downarrow_{k_3} v_4"']                         & \cdots \arrow[r, dashed, "F_m\, v_m \Downarrow_{k_m} v_{m+1}"']                         & \Delta_{m+1}
\end{tikzcd}
\end{center}

In a persistent cost model, we can thus reason about persistent usage
using sequential reasoning. We do not have to consider how the data structure
may change in between operations, since our reasoning can be applied to any accessible state automatically.
In particular, our reasoning does not have to consider monotonicity explicitly,
since the evaluation itself is guaranteed to be monotonic.

\begin{lemma}[\cite{Pilkiewicz:monotonic}]
\label{lem:real-persistent}
The real cost model $((\Downarrow^{\mathrm{R}}, \Phi^{\mathrm{R}}, \stripdot^{\mathrm{R}}), \sqsubseteq^{\mathrm{R}})$ is persistent.
\end{lemma}
\begin{minipage}{0.65\textwidth}
\begin{proof}
We prove this using induction.
The interesting case is if the evaluation forces a thunk,
which is also forced in the larger heap.
Then the evaluation succeeds using the (recall) rule,
since thunk evaluation is deterministic.
The full proof can be found in the appendix.
\end{proof}
\end{minipage}
\hfill
\begin{minipage}{0.33\textwidth}
\begin{center}
\begin{tikzcd}[row sep=large, column sep=5em, every label/.append style={scale=1.3}]
\Gamma \arrow[d, "{\text{[force]}}"] \arrow[r, "{\text{(force)}}"] & \Gamma' \arrow[d, dashed, "{\text{[refl]}}"] \\
\Delta \arrow[r, dashed, "{\text{(recall)}}"] & \Delta'
\end{tikzcd}
\end{center}
\end{minipage}
\vspace{0.1cm}

\begin{remark}
To show that the real cost model is persistent, we need to assume without loss of generality that we never allocate two values under the same name.
For example, we can allocate $\fold\, (\inl v)$ at $a$ if the heap is $\emptyset$,
but this fails if the heap is $\{a \mapsto \fold\, (\inr w)\}$, even though $\emptyset \sqsubseteq \{a \mapsto \fold\, (\inl v)\}$.
As such, the real cost model is only persistent if name clashes do not occur.

We could fix this issue in a more formal way by maintaining an explicit substitution in the definition of persistence.
A cost model would be persistent if for all $\Gamma : e \Downarrow_n \Gamma' : v$ and $\Gamma \sqsubseteq \Delta$,
there exist $\Delta', k, w$ and a substitution of pointers $\sigma$ 
such that $\Delta : e \Downarrow_k \Delta' : w$, $k \leq n$, $\sigma(v) = w$, $\sigma(\Gamma) = \Gamma$ and $\sigma(\Gamma') \sqsubseteq \Delta'$.
\end{remark}

\subsection{The Bankers Method, Revisited}

The Bankers Method is not persistent. To see why, consider how the heaps may change during evaluation.
Because we can both save and spend credits,
our accessibility relation has to account for arbitrary changes in the number of credits:

\begin{center}
  \AxiomC{$0 \leq n \leq m$}
  \RightLabel{[save]}
  \UnaryInfC{$\Gamma, a \mapsto_n hv \sqsubseteq \Gamma, a \mapsto_m hv$}
  \DisplayProof
  \qquad
  \AxiomC{$n \geq m \geq 0$}
  \RightLabel{[spend]}
  \UnaryInfC{$\Gamma, a \mapsto_n hv \sqsubseteq \Gamma, a \mapsto_m hv$}
  \DisplayProof
\end{center}

\begin{lemma}
\label{lem:bankers-accessible}
The Bankers cost model $((\Downarrow^{\mathrm{B}}, \Phi^{\mathrm{B}}, \stripdot^{\mathrm{B}}), \sqsubseteq^{\mathrm{B}})$ is accessible.
\end{lemma}

\begin{proposition}[\cite{Okasaki:purefun}]
The Bankers cost model $((\Downarrow^{\mathrm{B}}, \Phi^{\mathrm{B}}, \stripdot^{\mathrm{B}}), \sqsubseteq^{\mathrm{B}})$ is not persistent.
\end{proposition}
\begin{minipage}{0.65\textwidth}
\begin{proof}
Define $\Gamma = \{a \mapsto_5 \fold v\}$, $\Gamma' = \{a \mapsto_0 \fold v\}$ and $e = \spend{5}{a}{()}$.
Then $\Gamma : e \Downarrow^{\mathrm{B}}_0 \Gamma' : ()$.
Define $\Delta = \{a \mapsto_{0} \fold v\}$. Then $\Gamma \sqsubseteq \Delta$ by the [spend] rule.
But $\Delta : e$ is stuck, since the (spend) rule only applies if there are enough credits.
\end{proof}
\end{minipage}
\hfill
\begin{minipage}{0.33\textwidth}
\begin{center}
\begin{tikzcd}[row sep=large, column sep=5em, every label/.append style={scale=1.3}]
\Gamma \arrow[d, "{\text{[spend]}}"] \arrow[r, "{\text{(spend)}}"] & \Gamma' \\
\Delta &
\end{tikzcd}
\end{center}
\end{minipage}
\vspace{0.1cm}

This proposition shows the essential problem with the Bankers method in a persistent setting.
In order to perform amortised analysis of a data structure, we need to establish an invariant
on the number of credits it contains. But in a persistent setting, an operation on a
different version of the data structure may spend the credits and invalidate these invariants.

\section{Persistent Amortised Analysis with Credits}
\label{sec:credits}

The traditional Bankers method does not work in a persistent setting,
since it allows us to save and spend credits in an arbitrary way.
Okasaki~\cite{Okasaki:purefun} proposes to fix this issue by augmenting data structures with thunks.
We discuss his reasoning style in the next section.
Beforehand, we want to discuss a simpler reasoning style that is inspired by Danielsson's variant of debit passing style~\cite{Danielsson:amortisation}.
In this style, credits are saved only on lazy thunks, and may only be spent on the update when the thunk is forced~\cite{Pilkiewicz:monotonic}.
In this work, we omit all references to debits and consequently call it credit passing style.

Why does this approach allow for persistent amortised analysis?
One way to understand this is in terms of \textit{linearity}: the restriction that credits may not be duplicated.
The traditional analysis fails because it duplicates credits when spending them on an operation,
and maintains them on the previous version of the data structure at the same time.
In contrast, lazy thunks keep their content linear:
even if there are several references to the thunk, it is only forced once and the credits on it are only spent once.
If credits are only saved on lazy thunks, they stay linear,
even if the thunks themselves are used non-linearly.

\subsection{Credit Passing Style}

While the Bankers Method allows credit annotations on all heap values,
credit passing style only allows them on lazy thunks in the heap:
\begin{align*}
&\Gamma ::= \emptyset \mid \Gamma, a \mapsto hv \mid \Gamma, a \mapsto_n \lazy_F\; v
\end{align*}

The (memo) and (recall) rules remain unchanged, but we need to modify the rules for lazy thunks.
When we force a lazy thunk, we now take the credits stored on the thunk to pay for the computation.
If the thunk does not have enough credits, the evaluation can get stuck.
To ensure that enough credits are available, we add a (save) rule, which saves credits on a lazy thunk.
If we save credits on a memoised thunk, they will be wasted.

\begin{center}
    \AxiomC{$F \in \mathcal{F}$}
    \AxiomC{$a \not\in \mathrm{dom}(\Gamma)$}
    \RightLabel{(lazy)}
    \BinaryInfC{$\Gamma : \lazy_F\; v \Downarrow_0 (\Gamma, a \mapsto_0 \lazy_F\; v) : a$}
    \DisplayProof
    \qquad
    \AxiomC{$a \mapsto \memo\; v \in \Gamma$}
    \RightLabel{(waste)}
    \UnaryInfC{$\Gamma : \save{m}{a} \Downarrow_{m} \Gamma : a$}
    \DisplayProof
\end{center}
\begin{center}
    \AxiomC{\phantom{$x$}}
    \RightLabel{(save)}
    \UnaryInfC{$(\Gamma, a \mapsto_n \lazy_F\; v) : \save{m}{a} \Downarrow_{m} (\Gamma, a \mapsto_{n+m} \lazy_F\; v) : a$}
    \DisplayProof
\end{center}
\begin{center}
    \AxiomC{$\Gamma : F\, v \Downarrow_k \Delta : w$}
    \AxiomC{$k \leq n$}
    \RightLabel{(force)}
    \BinaryInfC{$(\Gamma, a \mapsto_n \lazy_F\; v) : \mathrm{force}\; a \Downarrow_{0} (\Delta, a \mapsto \memo\; w) : w$}
    \DisplayProof
\end{center}

We call this the credit passing semantics and write $\Downarrow^{\mathrm{C}}_k$ to refer to it.
The potential function $\Phi^{\mathrm{C}}$ counts the total number of credits in the heap
and the $\stripdot^{\mathrm{C}}$ function removes the credit annotations:
\begin{equation*}
  \Phi^{\mathrm{C}}(\emptyset) = 0 \quad\quad \Phi^{\mathrm{C}}(\Gamma, a \mapsto_n \lazy_F\; v) = \Phi^{\mathrm{C}}(\Gamma) + n \quad\quad \Phi^{\mathrm{C}}(\Gamma, a \mapsto hv) = \Phi^{\mathrm{C}}(\Gamma)
\end{equation*}
\begin{equation*}
  \strip{\emptyset}^{\mathrm{C}} = \emptyset \quad\quad \strip{\Gamma, a \mapsto_n \lazy_F\; v}^{\mathrm{C}} = \strip{\Gamma}^{\mathrm{C}}, a \mapsto \lazy_F\; v \quad\quad \strip{\Gamma, a \mapsto hv}^{\mathrm{C}} = \strip{\Gamma}^{\mathrm{C}}, a \mapsto hv
\end{equation*}

\begin{lemma}
\label{lem:credit-passing-sound}
  The credit passing cost model $(\Downarrow^{\mathrm{C}}, \Phi^{\mathrm{C}}, \stripdot^{\mathrm{C}})$ is sound.
\end{lemma}

Let us consider how the heaps evolve in the credit passing semantics.
Our accessibility relation $\sqsubseteq^{\mathrm{C}}$ extends $\sqsubseteq^{\mathrm{R}}$ by the rules:

\begin{center}
  \AxiomC{$hv = \lazy_F\; v$}
  \AxiomC{$n \leq m$}
  \RightLabel{[save]}
  \BinaryInfC{$(\Gamma, a \mapsto_n hv) \sqsubseteq (\Gamma, a \mapsto_m hv)$}
  \DisplayProof
  \quad
  \AxiomC{$\Gamma \sqsubseteq \Delta$}
  \AxiomC{$\Gamma : F\, v \Downarrow_k \Delta : w$}
  \AxiomC{$k \leq n$}
  \RightLabel{[force]}
  \TrinaryInfC{$(\Gamma, a \mapsto_n \lazy_F\; v) \sqsubseteq (\Delta, a \mapsto \memo\; w)$}
  \DisplayProof
\end{center}

Unlike in the traditional Bankers Method,
the number of credits on a \textit{lazy} thunk may only increase in this semantics.
It can decrease when the thunk is forced,
but at that point the thunk is replaced by a memoised value,
which costs nothing to force.
This ensures that the heaps can only become ``better'' over time:

\begin{lemma}
\label{lem:credit-passing-persistent}
The credit passing cost model $((\Downarrow^{\mathrm{C}}, \Phi^{\mathrm{C}}, \stripdot^{\mathrm{C}}), \sqsubseteq^{\mathrm{C}})$ is uniformly persistent.
\end{lemma}
\begin{proof} If we force a thunk on which more credits have been saved,
  those credits are discarded. If we save credits on a thunk that has been forced,
  we waste these credits.
\begin{center}
\begin{tikzcd}[row sep=large, column sep=7em, every label/.append style={scale=1.3}]
\Gamma \arrow[d, "{\text{[save]}}"] \arrow[r, "{\text{(force)}}"] & \Gamma' \arrow[d, dashed, "{\text{[refl]}}"] \\
\Delta \arrow[r, dashed, "{\text{(force)}}"] & \Delta'
\end{tikzcd}
\qquad
\begin{tikzcd}[row sep=large, column sep=7em, every label/.append style={scale=1.3}]
\Gamma \arrow[d, "{\text{[force]}}"] \arrow[r, "{\text{(save)}}"] & \Gamma' \arrow[d, dashed, "{\text{[force]}}"] \\
\Delta \arrow[r, dashed, "{\text{(waste)}}"] & \Delta'
\end{tikzcd}
\end{center}
\end{proof}

\subsection{Persistent Binary Counter}

To illustrate the credit passing technique, we add thunks to the binary counter.
We consider thunks as internal nodes that correspond to a delayed increment operation.
When we force a thunk, we perform one step of the increment operation:

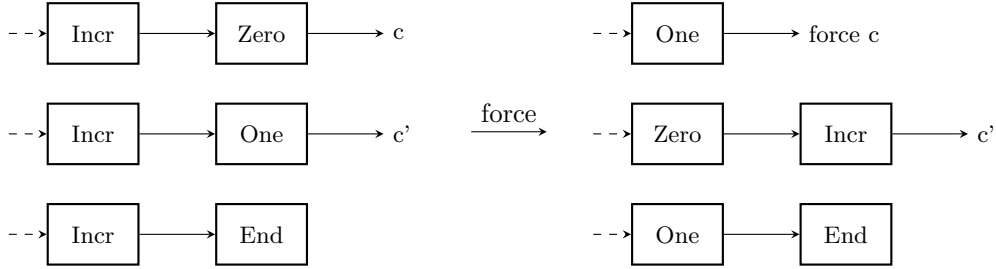
\vspace{0.2cm}
\begin{center}
\begin{tikzpicture}[
  node distance=0.7cm and 1cm,
  counter node/.style={rectangle, draw=black, thick, minimum width=1.2cm, minimum height=0.8cm, font=\small},
  label node/.style={font=\small},
  >=stealth,
]

  % Left binary counter (11001) - vertical
  \node[counter node] (c2_1) {Incr};
  \node[label node, left=0.5cm of c2_1] (c2_0) {};
  \node[counter node, right=1cm of c2_1] (c2_2) {Zero};
  \node[label node, right=1cm of c2_2] (c2_3) {c};
  \node[counter node, below=0.5cm of c2_1] (c3_1) {Incr};
  \node[label node, left=0.5cm of c3_1] (c3_0) {};
  \node[counter node, right=1cm of c3_1] (c3_2) {One};
  \node[label node, right=1cm of c3_2] (c3_3) {c'};
  \node[counter node, below=0.5cm of c3_1] (c1_1) {Incr};
  \node[label node, left=0.5cm of c1_1] (c1_0) {};
  \node[counter node, right=1cm of c1_1] (c1_2) {End};

  % Links between left counter nodes
  \draw[->, dashed] (c1_0) -- (c1_1);
  \draw[->, dashed] (c2_0) -- (c2_1);
  \draw[->, dashed] (c3_0) -- (c3_1);
  \draw[->] (c1_1) -- (c1_2);
  \draw[->] (c2_1) -- (c2_2);
  \draw[->] (c2_2) -- (c2_3);
  \draw[->] (c3_1) -- (c3_2);
  \draw[->] (c3_2) -- (c3_3);

  % Transformation arrow
  \draw[->] (5, -1.3) -- (6, -1.3) node[midway, above=0cm] {force};

  % Right binary counter (00101) - vertical
  \node[counter node, right=6.5cm of c2_1] (c5_1) {One};
  \node[label node, left=0.5cm of c5_1] (c5_0) {};
  \node[label node, right=1cm of c5_1] (c5_2) {force c};
  \node[counter node, below=0.5cm of c5_1] (c6_1) {Zero};
  \node[label node, left=0.5cm of c6_1] (c6_0) {};
  \node[counter node, right=1cm of c6_1] (c6_2) {Incr};
  \node[label node, right=1cm of c6_2] (c6_3) {c'};
  \node[counter node, below=0.5cm of c6_1] (c4_1) {One};
  \node[label node, left=0.5cm of c4_1] (c4_0) {};
  \node[counter node, right=1cm of c4_1] (c4_2) {End};
  
  % Links between right counter nodes
  \draw[->, dashed] (c5_0) -- (c5_1);
  \draw[->, dashed] (c6_0) -- (c6_1);
  \draw[->, dashed] (c4_0) -- (c4_1);
  \draw[->] (c4_1) -- (c4_2);
  \draw[->] (c5_1) -- (c5_2);
  \draw[->] (c6_1) -- (c6_2);
  \draw[->] (c6_2) -- (c6_3);

\end{tikzpicture}
\end{center}

Again, we define the heaps of binary counters inductively.
This time, each zero-bit is followed by a lazy increment thunk.
We assign one credit to increment thunks, but not to other heap allocations.
\begin{align*}
&\mathrm{Counter}(a \mapsto \fold \texttt{End} : a) & \\
\mathrm{Counter}(\Gamma : a) \implies &\mathrm{Counter}(\Gamma, b \mapsto \fold \texttt{Cons}(\texttt{One}, a) : b) & (b \notin \mathrm{dom}(\Gamma)) \\
\mathrm{Counter}(\Gamma : a) \implies &\mathrm{Counter}(\Gamma, c \mapsto \fold \texttt{Cons}(\texttt{Zero}, b), b \mapsto_1 \lazy_{\texttt{Incr}}\; a : c) & (b,c \notin \mathrm{dom}(\Gamma))
\end{align*}

We change the increment operation as follows:
\begin{align*}
&\mathrm{incr}(c) = \force (\save{2}{(\lazy_{\texttt{Incr}}\;c)}) \\
&\texttt{Incr}(c) = \case (\unfold c) \{ \\
&\quad \texttt{End} \mapsto \fold \texttt{Cons}(\texttt{One}, c) \\
&\quad \texttt{Cons}(n, a) \mapsto \case n \{\\
&\quad\quad \texttt{Zero} \mapsto \fold \texttt{Cons}(\texttt{One}, \force (\save{1}{a})) \\
&\quad\quad \texttt{One}  \mapsto \fold \texttt{Cons}(\texttt{Zero}, \save{1}{(\lazy_{\texttt{Incr}}\;a)}) \,\} \,\}
\end{align*}

To increment a counter $c$, we just add an $\texttt{Incr}(c)$ thunk to the front,
save two credits on it and force it. To show that this works, we have to prove that the
thunk can be forced in two steps and the resulting counter is still well-formed.

\begin{lemma}
\label{lem:counter-credits}
If $\mathrm{Counter}(\Gamma : c)$, then $\Gamma : \texttt{Incr}(c) \Downarrow^{\mathrm{C}}_k \Delta : c'$ with $\mathrm{Counter}(\Delta : c')$ and $k \leq 2$.
\end{lemma}
\begin{proof}
By induction on the $\mathrm{Counter}$ predicate. In each case, the function call takes one step.
\begin{itemize}
  \bitem{(End)}
    From $\mathrm{Counter}(\Gamma : c)$, we obtain
    $\mathrm{Counter}(\Gamma, c' \mapsto \fold \texttt{Cons}(\texttt{One}, c) : c')$,
    for a total cost of $1$.
  \bitem{(Zero)} 
    From $\mathrm{Counter}(\Gamma, c \mapsto \fold \texttt{Cons}(\texttt{Zero}, a), a \mapsto_1 \lazy_{\texttt{Incr}}\; b : c)$,
    we deduce $\mathrm{Counter}(\Gamma : b)$.
    We save a credit on $a$ for a total of two credits. This allows us to force the delayed
    $\texttt{Incr(b)}$ computation, which yields $\mathrm{Counter}(\Delta : b')$ by the induction hypothesis.
    We obtain $\mathrm{Counter}(\Delta, c' \mapsto \fold \texttt{Cons}(\texttt{One}, b') : c')$.
    Because we save one credit, the total cost is $2$.
  \bitem{(One)}
    From $\mathrm{Counter}(\Gamma, c \mapsto \fold \texttt{Cons}(\texttt{One}, a) : c)$,
    we deduce $\mathrm{Counter}(\Gamma : a)$. We allocate a new thunk and save a credit on it,
    thus obtaining $\mathrm{Counter}(\Gamma, c' \mapsto \fold \texttt{Cons}(\texttt{Zero}, b), b \mapsto_1 \lazy_{\texttt{Incr}}\; a : c')$.
    Because we save one credit, the total cost is $2$.
\end{itemize}
\end{proof}

In the persistent binary counter, every lazy thunk contains an $\texttt{Incr}$ operation
with a certain number of credits on it.
In more advanced data structures, there are usually thunks for different operations
with different numbers of credits on them.
For example, a persistent double-ended queue would contain thunks
for delayed `cons', `tail', `snoc', and `init' operations.
This can be supported in our framework by annotating thunks $\lazy_{F_i}\; v$
with different operations $F_1, \ldots, F_n$.
Of course, if different operations $F_i$ have different costs, then the reasoning has
to accommodate that --- this is a frequent source of complexity in Okasaki's book.

\subsection{Credit Inheritance}

Lorenzen~\cite{Lorenzen:creditmonad} proposes credit inheritance to analyse Okasaki's Bankers Queue.
Credit passing style cannot be used to analyse that data structure,
as it wastes credits that are saved on a memoised thunk.
To solve this issue, credit inheritance allows each thunk to designate another thunk as its \textit{heir}.
If a thunk has an excess of credits when forced, it may pass excess credits on to its heir instead of wasting them.
Similarly, once a thunk is memoised, all credits saved on it are passed on to the heir as well.

We extend the syntax of heaps to add an heir $h$ on memoised thunks:
\begin{align*}
&\Gamma ::= \emptyset \mid \Gamma, a \mapsto hv \mid \Gamma, a \mapsto_n \lazy_F\; v \mid \Gamma, a \mapsto_h \memo\; v
\end{align*}

Similarly, we extend all existing rules to propagate the heir.
Rules that do not have the judgement as a premise (like (value) or (unfold))
use the empty heir $\emptyset$. The other rules propagate the heir from their premise.
Since the (let) rule has two premises, we add a side-condition that the premises do not choose different heirs:

\begin{center}
  \AxiomC{$\Gamma : e_1 \Downarrow^{h_1}_k \Delta : v$}
  \AxiomC{$\Delta : e_2[v/x] \Downarrow^{h_2}_l \Theta : w$}
  \AxiomC{$|h_1 \cup h_2| \leq 1$}
  \RightLabel{(let)}
  \TrinaryInfC{$\Gamma : \llet x = e_1 \lin e_2 \Downarrow^{h_1 \cup h_2}_{k+l} \Theta : w$}
  \DisplayProof
\end{center}

The key new rule is the (pass) rule, which allows us to designate another thunk as the heir $h$
and pass on credits from our computation to the heir.
We record the heir on the reduction relation $\Downarrow^{\{h\}}_k$ of the current computation.
This rule does not specify how many credits $m$ are passed on to the heir;
this will be determined later by the (force) rule.

\begin{center}
  \AxiomC{$\Gamma : \save{m}{h} \Downarrow^\emptyset_m \Delta : v$}
  \RightLabel{(pass)}
  \UnaryInfC{$\Gamma : \pass{h} \Downarrow^{\{h\}}_m \Delta : v$}
  \DisplayProof
\end{center}

The (force) rule installs the heir annotation on the memoised thunk
and removes it from the judgement. This time we ask that the computation
in the thunk spends \textit{all} the credits available on the thunk.
This ensures that the number $m$ chosen in (pass) is as large as possible
and no excess credits are wasted.

\begin{center}
    \AxiomC{$\Gamma : F\, v \Downarrow^{\{h\}}_k \Delta : w$}
    \AxiomC{$k = n$}
    \RightLabel{(force)}
    \BinaryInfC{$(\Gamma, a \mapsto_n \lazy_F\; v) : \mathrm{force}\; a \Downarrow^\emptyset_{0} (\Delta, a \mapsto_h \memo\; w) : w$}
    \DisplayProof
\end{center}

The (pass) rule allows us to choose the number of credits $m$ to pass on seemingly without constraint.
However, the requirement that the (force) rule spends all credits implies that $m$ has to be the difference
between the number of credits and the cost of the computation.
It is unfortunate that the rules enforce this invariant only indirectly,
but this seems unavoidable in a big-step semantics. In a small-step semantics,
the (pass) rule would have direct access to the number of credits that are left on the thunk
and could choose $m$ deterministically.
This is done in the testing framework of Lorenzen~\cite{Lorenzen:creditmonad}.

Finally, the heir annotation allows us to provide the (inherit) rule.
Unlike the (waste) rule, which discards credits placed on memoised thunks,
the (inherit) rule allows us to retain these credits by passing them on to the heir.

\begin{center}
  \AxiomC{$\Gamma : \save{m}{h} \Downarrow^\emptyset_m \Delta : v$}
  \RightLabel{(inherit)}
  \UnaryInfC{$(\Gamma, a \mapsto_h \memo\; w) : \save{m}{a} \Downarrow^\emptyset_m (\Delta, a \mapsto_h \memo\; w) : a$}
  \DisplayProof
\end{center}

We call this the credit-inheritance semantics and write $\Downarrow^{\mathrm{CI}}_k$ for $\Downarrow^\emptyset_k$.
The potential function $\Phi^{\mathrm{CI}} = \Phi^{\mathrm{C}}$ counts the total number of credits in the heap
and the $\stripdot^{\mathrm{CI}}$ function also removes the heir annotations ($\strip{\Gamma, a \mapsto_h \memo\; v}^{\mathrm{CI}} = \strip{\Gamma}^{\mathrm{CI}}, a \mapsto \memo\; v$).
In the real semantics, we interpret `$\pass{h}$' as a no-op that returns $h$ without changing the heap.

\begin{lemma}
\label{lem:credit-inheritance-sound}
  The credit-inheritance cost model $(\Downarrow^{\mathrm{CI}}, \Phi^{\mathrm{CI}}, \stripdot^{\mathrm{CI}})$ is sound.
\end{lemma}

The $\sqsubseteq^{\mathrm{CI}}$ relation differs from $\sqsubseteq^{\mathrm{C}}$ only in the new heir annotation in the [force] rule.

\begin{lemma}
\label{lem:credit-inheritance-persistent}
The credit-inheritance cost model $((\Downarrow^{\mathrm{CI}}, \Phi^{\mathrm{CI}}, \stripdot^{\mathrm{CI}}), \sqsubseteq^{\mathrm{CI}})$ is uniformly persistent.
\end{lemma}
\begin{proof}
  If we force a thunk on which more credits have been saved,
  those extra credits will be passed on to the heir instead of being discarded.
  If we save credits on a thunk that has already been forced,
  those credits will be inherited by the heir instead of being wasted.
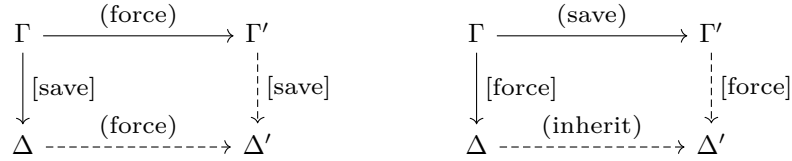
\begin{center}
\begin{tikzcd}[row sep=large, column sep=7em, every label/.append style={scale=1.3}]
\Gamma \arrow[d, "{\text{[save]}}"] \arrow[r, "{\text{(force)}}"] & \Gamma' \arrow[d, dashed, "{\text{[save]}}"] \\
\Delta \arrow[r, dashed, "{\text{(force)}}"] & \Delta'
\end{tikzcd}
\qquad\qquad
\begin{tikzcd}[row sep=large, column sep=7em, every label/.append style={scale=1.3}]
\Gamma \arrow[d, "{\text{[force]}}"] \arrow[r, "{\text{(save)}}"] & \Gamma' \arrow[d, dashed, "{\text{[force]}}"] \\
\Delta \arrow[r, dashed, "{\text{(inherit)}}"] & \Delta'
\end{tikzcd}
\end{center}
\end{proof}

Credit inheritance is loosely inspired by Okasaki's debit inheritance~\cite{Okasaki:purefun}.
However, it usually leads to a different style of reasoning.
For example, Lorenzen~\cite[Section~6.6]{Lorenzen:creditmonad} uses it to analyse the Bankers Queue.
In that data structure, the `append' operation creates a long list of thunks, which all need a credit to run.
This is an ideal setting for credit inheritance,
since each thunk can name its successor in the list as its heir.
This makes it possible to store credits on the head of the list,
from which the (inherit) rule propagates them to the first lazy thunk.
To reason about the queue, we thus do not need to know at
which index of the list the first lazy thunk can be found;
as long as the amount of credits saved on the head equals the length of the list,
the computation succeeds.

Credit inheritance also shows the power of using credits rather than debits.
We would not be able to perform the same analysis if we simply replaced credits by debits.
Debits require us to know the cost of updating a thunk when creating it.
In contrast, credit inheritance allows us to let the cost of updating a thunk grow
dynamically to match the number of credits we have:
all excess credits are simply passed on to the heir.
It is possible to emulate this in a debit-based analysis,
but we would have to wait until all debits are paid off,
and then change the thunk retroactively to include
further save-statements that increase its debits~\cite{Pottier:thunks}.

\section{Persistent Amortised Analysis with Debits}
\label{sec:debits}

Okasaki does not reason in terms of credits, but uses \textit{debits}:
a number that describes how many credits we have to save on a thunk before we can force it.
One may view a debit as just a negative credit:
if the evaluation of a thunk takes $k$ steps and the thunk has $n$ credits,
we can equivalently say that the thunk has $k - n$ debits.
But as we have seen in the last section, one can perform
persistent amortised analysis purely in terms of credits.
If debits are not necessary,
then why does Okasaki's work use debits at all?

We believe that Okasaki's use of debits should be seen as part of his evaluation model of thunks.
In his analysis, Okasaki acts as if thunks were evaluated upon creation.
But he does not pay for the evaluation and instead places the cost as a debit on the result of the thunk.
Debits thus act as a barrier that prevents accesses to the result until the evaluation has been paid for.
Okasaki motivates debits by analogy to a ``layaway plan'',
where you select what to buy and the store holds it for you
until you have fully paid for it~\cite[page~60]{Okasaki:purefun}.

Debits are a natural choice when we assume that thunks are evaluated upon creation.
In that setting, we know the cost of updating the thunk.
We could still use credits and count up until we have enough credits to cover the cost,
but it is much more elegant to use debits and count down until we reach zero remaining debits instead.

\subsection{The Persistent Bankers Method}

Okasaki defines the amortised cost of an operation informally
as the \textit{unshared} cost of actually performing work
plus the number of debits discharged~\cite[page~60]{Okasaki:purefun}.
To formalise this, we track two numbers $\Downarrow_{k,k'}$,
where $k$ is the total number of time steps taken by the evaluation
and $k'$ is the number of debits discharged during the evaluation.
All previous rules can be extended to this setting by either propagating $k'$
from their single premise to their conclusions, returning $k' = 0$ if no premise exists,
or, in the (let) rule,  adding the $k'$ numbers of the two premises.

In our heaps, we now allow debits to be placed on memoised thunks. Unlike credits, debits may be negative.
We write $\memo\; w$ for a memoised thunk that has been accessed by the program
and thus holds no debits.
For a memoised thunk that has not yet been accessed,
we write $\memo^{\Delta}_{F\,v}\; w$ to record the computation $F\,v$ that computes $w$,
as well as the new allocations $\Delta$ created by the computation.
These additional records are not necessary for reasoning, but enable us to connect this semantics to the real semantics later.
Our syntax for heaps is now:
\begin{align*}
&\Gamma ::= \emptyset \mid \Gamma, a \mapsto hv \mid \Gamma, a \mapsto_n \memo^{\Delta}_{F\,v}\; w
\end{align*}

In the (lazy) rule, we now run the computation of the thunk immediately, but do not pay for it.
Instead, we place the unshared cost as a debit on the resulting thunk.
However, if the thunk itself pays off debits, this cost is propagated to the outside.
The (force) rule allows us to access the value of a memoised thunk
once the debt has been paid off. Once we access it, we remove the annotations from the thunk.
Finally, the (save) rule allows us to pay off debits on a memoised thunk
and if the thunk has no debits left, saved credits end up being wasted.

\begin{center}
    \AxiomC{$\Gamma : F\, v \Downarrow_{k,k'} \Delta : w$}
    \AxiomC{$a \not\in \mathrm{dom}(\Delta)$}
    \RightLabel{(lazy)}
    \BinaryInfC{$\Gamma : \lazy_F\; v \Downarrow_{0,k'} (\Delta, a \mapsto_{k} \memo^{\Delta\setminus\Gamma}_{F\,v}\; w) : a$}
    \DisplayProof
\end{center}
\begin{center}
    \AxiomC{$n \leq 0$}
    \RightLabel{(force)}
    \UnaryInfC{$\Gamma, a \mapsto_n \memo^{\Delta}_{F\,v}\; w : \mathrm{force}\; a \Downarrow_{0,0} \Gamma, a \mapsto \memo\; w : w$}
    \DisplayProof
\end{center}
\begin{center}
    \AxiomC{$hv = \memo^{\Delta}_{F\,v}\; w$}
    \RightLabel{(save)}
    \UnaryInfC{$(\Gamma, a \mapsto_n hv) : \save{m}{a} \Downarrow_{0,m} (\Gamma, a \mapsto_{n-m} hv) : a$}
    \DisplayProof
\end{center}
\begin{center}
    \AxiomC{$a \mapsto \memo\; w \in \Gamma$}
    \RightLabel{(waste)}
    \UnaryInfC{$\Gamma : \save{m}{a} \Downarrow_{0,m} \Gamma : a$}
    \DisplayProof
\end{center}

We call this the debit semantics
and write $\Downarrow^{\mathrm{D}}_{k+k'}$ for $\Downarrow_{k,k'}$.
To show that this semantics is sound, we need to undo the evaluation performed in the (lazy) rule.
If a thunk still has an annotation, we restore it to the lazy state
and remove all new allocations that were added during the computation.
In the potential function, we count the number of debits that were paid off,
where $k^{\mathrm{R}}_{F v}$ is the real cost of evaluating the thunk and $n$ is its number of debits:
\begin{align*}
  &\Phi^{\mathrm{D}}(\emptyset) = 0 \quad\quad \Phi^{\mathrm{D}}(\Gamma, a \mapsto_n \memo^{\Delta}_{F\,v}\; w) = \Phi^{\mathrm{D}}(\Gamma) + k^{\mathrm{R}}_{F v} - n \quad\quad \Phi^{\mathrm{D}}(\Gamma, a \mapsto hv) = \Phi^{\mathrm{D}}(\Gamma) \\
  &\strip{\emptyset}^{\mathrm{D}} = \emptyset \quad \strip{\Gamma, a \mapsto_n \memo^{\Delta}_{F\,v}\; w}^{\mathrm{D}} = \strip{\Gamma\setminus\Delta}^{\mathrm{D}}, a \mapsto \lazy_F\; v \quad \strip{\Gamma, a \mapsto hv}^{\mathrm{D}} = \strip{\Gamma}^{\mathrm{D}}, a \mapsto hv
\end{align*}

\begin{lemma}
\label{lem:debit-sound}
The debit cost model $(\Downarrow^{\mathrm{D}}, \Phi^{\mathrm{D}}, \stripdot^{\mathrm{D}})$ is sound.
\end{lemma}

Let us consider how the heaps evolve in the debit semantics.
Our accessibility relation $\sqsubseteq^{\mathrm{D}}$ on heaps is defined as follows:

\begin{center}
  \AxiomC{$\Gamma \sqsubseteq \Delta$}
  \AxiomC{$\Gamma : F\, v \Downarrow_{k,k'} \Delta : w$}
  \AxiomC{$a \notin \mathrm{dom}(\Delta)$}
  \RightLabel{[lazy]}
  \TrinaryInfC{$\Gamma \sqsubseteq (\Delta, a \mapsto_k \memo^{\Delta\setminus\Gamma}_{F\,v}\; w)$}
  \DisplayProof
\end{center}
\begin{center}
  \AxiomC{$hv = \memo^{\Delta}_{F\,v}\; w$}
  \AxiomC{$n \geq m$}
  \RightLabel{[save]}
  \BinaryInfC{$(\Gamma, a \mapsto_n hv) \sqsubseteq (\Gamma, a \mapsto_m hv)$}
  \DisplayProof
  \quad
  \AxiomC{$n \leq 0$}
  \RightLabel{[force]}
  \UnaryInfC{$(\Gamma, a \mapsto_n \memo^{\Delta}_{F\,v}\; w) \sqsubseteq (\Gamma, a \mapsto \memo\; w)$}
  \DisplayProof
\end{center}

Again, the number of debits on memoised thunks may only decrease in this semantics.
This ensures that the heaps can only become ``better'' over time:

\begin{lemma}
\label{lem:debit-persistent}
The debit cost model $((\Downarrow^{\mathrm{D}}, \Phi^{\mathrm{D}}, \stripdot^{\mathrm{D}}), \sqsubseteq^{\mathrm{D}})$ is uniformly persistent.
\end{lemma}

Perhaps the most surprising aspect of the debit semantics is that it has to track both $k$ and $k'$,
the unshared cost of evaluation and the number of debits discharged.
This is necessary to split these costs in the (lazy) rule, where the cost $k'$
may not be placed as a debit on the resulting thunk, but has to be propagated to the outside instead.
Indeed, for a top-level computation that lives outside of any thunks, this distinction does not matter
and we can specify the total cost as $\Downarrow^{\mathrm{D}}_{k+k'}$ for $\Downarrow_{k,k'}$.

\begin{example}
\label{ex:tracking-both}
The debit semantics would be unsound if the (lazy) rule placed $k'$ as a debit on the resulting thunk.
Assume we are given a heap $\Gamma = \{a \mapsto_n \memo_{{F_1}\; v}\; w\}$
and the expression $e = \lazy_{F_2} ()$ for $F_2(()) = \save{n}{a}$.
Our correct (lazy) rule allows us to derive:
\[
\Gamma : e \Downarrow_{0,n} \Delta : b$ with $\Delta = \{a \mapsto_0 \memo_{F_1\; v}\; w, b \mapsto_1 \memo^\emptyset_{F_2\; ()}\; a\}
\]
That is, the $\save{n}{a}$ operation is performed, which removes the debits from $a$, and we propagate the cost $n$ to the outside.
Instead, the unsound (lazy) rule places the discharged debits on the resulting thunk and allows us to derive:
\[
\Gamma : e \Downarrow_{0,0} \Delta : b$ with $\Delta = \{a \mapsto_0 \memo_{F_1\; v}\; w, b \mapsto_{n+1} \memo^\emptyset_{F_2\; ()}\; a\}
\]
This derivation does not propagate the cost $n$ to the outside,
but simply moves $n$ debits from $a$ to $b$.
Since $a$ now has no remaining debits, we may force it, without forcing $b$ at all.
This allows us to access $w$ at no cost and thus breaks soundness.
\end{example}

\subsection{Persistent Binary Counter}

To illustrate Okasaki's reasoning, we analyse the binary counter using debits. 
Again, we define the heaps of counters inductively.
We assume that all thunks are memoised, but place one debit on memoised thunks.
For simplicity, we omit the annotations $F\, v, \Delta$ on unaccessed memoised thunks.
\begin{align*}
&\mathrm{Counter}(a \mapsto \fold \texttt{End} : a) & \\
\mathrm{Counter}(\Gamma : a) \implies &\mathrm{Counter}(\Gamma, b \mapsto \fold \texttt{Cons}(\texttt{One}, a) : b) & (b \notin \mathrm{dom}(\Gamma)) \\
\mathrm{Counter}(\Gamma : a) \implies &\mathrm{Counter}(\Gamma, c \mapsto \fold \texttt{Cons}(\texttt{Zero}, b), b \mapsto_1 \memo\; a : c) & (b,c \notin \mathrm{dom}(\Gamma))
\end{align*}

The implementation of the persistent binary counter stays largely the same as in \cref{sec:credits}.
The only difference is that we need to save one less credit on the increment thunks.
That is, we modify the $\mathrm{incr}$ operation
by replacing $\save{2}{(\lazy_{\texttt{Incr}}\;c)}$ with $\save{1}{(\lazy_{\texttt{Incr}}\;c)}$ and
the lazy $\texttt{Incr}$ function
by replacing $\save{1}{(\lazy_{\texttt{Incr}}\;a)}$ with $\lazy_{\texttt{Incr}}\;a$.
In our proof, we now assume that all thunks are evaluated upon creation
and place the cost of that evaluation as a debit on the resulting thunk.
As in the debit semantics, we account separately for the cost of paying off debits.

\begin{lemma}
\label{lem:counter-debits}
If $\mathrm{Counter}(\Gamma : c)$, then $\Gamma : \mathrm{Incr}(c) \Downarrow^{\mathrm{D}}_{1,k'} \Delta : c'$
with $\mathrm{Counter}(\Delta : c')$ and $k' \leq 1$.
\end{lemma}
\begin{proof}
By induction on the $\mathrm{Counter}$ predicate. In each case, the function call takes one step.
\begin{itemize}
  \bitem{(End)}
    From $\mathrm{Counter}(\Gamma : c)$, we obtain
    $\mathrm{Counter}(\Gamma, c' \mapsto \fold \texttt{Cons}(\texttt{One}, c) : c')$.
  \bitem{(Zero)} 
    From $\mathrm{Counter}(\Gamma, c \mapsto \fold \texttt{Cons}(\texttt{Zero}, b), b \mapsto_1 \memo\; a : c)$,
    we deduce $\mathrm{Counter}(\Gamma : a)$. We pay off the debit on $b$ and
    obtain $\mathrm{Counter}(\Gamma, c' \mapsto \fold \texttt{Cons}(\texttt{One}, a) : c')$.
  \bitem{(One)}
    From $\mathrm{Counter}(\Gamma, c \mapsto \fold \texttt{Cons}(\texttt{One}, a) : c)$,
    we deduce $\mathrm{Counter}(\Gamma : a)$.
    By the induction hypothesis, $\Gamma : \mathrm{Incr}(a) \Downarrow^{\mathrm{D}}_{1,1} \Delta : c'$
    with $\mathrm{Counter}(\Delta : c')$.
    We obtain $\mathrm{Counter}(\Delta, c'' \mapsto \fold \texttt{Cons}(\texttt{Zero}, b), b \mapsto_1 \memo\; c' : c'')$,
    where we place the cost of evaluation as a debit on $b$ and propagate the cost of paying off the debit.
\end{itemize}
\end{proof}

This proof is quite similar to the sequential one (\cref{lem:counter-sequentially}).
By pretending that $\texttt{Incr}$ thunks are evaluated immediately,
we can reason about them just as we reasoned about the \texttt{incr} function.
The main difference is that this proof places debits on zero-bits,
while the sequential proof places credits on one-bits.
In the recursive case, where one-bits are turned into zero-bits,
the sequential proof consumes credits while the persistent proof produces debits.

This proof is typical for Okasaki's debit method.
Perhaps surprisingly, this proof reasons about thunks that have not yet been created.
On a counter consisting only of one-bits, the program creates
an increment thunk and forces it to flip the first one-bit to a zero-bit.
However, in the proof, we act as if we continued evaluating the thunk
and flipped all one-bits to zero-bits immediately$\dots$
even though the program has not yet created the thunks for those flips.
This is possible because thunk evaluation is deterministic:
One can argue that ``in the future, thunk $a$ will create a new thunk $b$''
and then reason about $b$ before it actually exists.

\subsection{Debit Inheritance}

In the last section, we saw that it is generally unsound to move debits between thunks.
This is sound, however, if we can guarantee an evaluation order between the thunks.
For example, we might know that $b$ will always be forced before $a$.
Okasaki then allows $b$ to \textit{inherit} the debit from $a$~\cite[page~67]{Okasaki:purefun}.

An important special case, where it is guaranteed that $b$ will always be forced before $a$
is if $a$ is created during the evaluation of $b$.
Even though the (lazy) rule evaluates $b$ immediately,
the result of that computation (and thus $a$) only becomes accessible
once $b$ has been forced.
This ensures that $a$ is only accessible from outside the thunk
once $b$'s debit has been paid off.

To encode this in our semantics, we modify the (lazy) rule.
For a freshly allocated thunk, we allow propagating some of the unshared work to the outside.

\begin{center}
    \AxiomC{$\Gamma : F\, v \Downarrow_{(k_1+k_2),k'} \Delta : w$}
    \AxiomC{$a \not\in \mathrm{dom}(\Gamma)$}
    \RightLabel{(lazy)}
    \BinaryInfC{$\Gamma : \lazy_F\; v \Downarrow_{k_1,k'} (\Delta, a \mapsto_{k_2} \memo^{\Delta\setminus\Gamma}_{F\,v}\; w) : a$}
    \DisplayProof
\end{center}

Our modification may seem counter-intuitive:
rather than changing how thunks discharge debits,
we change how freshly allocated thunks distribute their unshared cost.
However, if $b$ is a thunk that creates $a$ during its evaluation,
then this (lazy) rule allows us to place the unshared cost of $a$ on $b$ instead.
Note that this cannot be emulated by the $(save)$ rule,
since $a$'s debits become part of the unshared cost $k$ of $b$,
instead of the discharge cost $k'$ of $b$.

This (lazy) rule captures Okasaki's reasoning style for \textit{monolithic} computations,
which create new thunks that have to be forced in the computation itself.
When reasoning about them, ``all debits are usually assigned to the root''~\cite[page~61]{Okasaki:purefun}.

We call this the debit-inheritance semantics
and write $\Gamma : e \Downarrow^{\mathrm{DI}}_{k+k'} \Delta : w$ for
$\Gamma : e \Downarrow_{k,k'} \Delta : w$.
We reuse the $\Phi^{\mathrm{D}}$ and $\stripdot^{\mathrm{D}}$ defined for the debit semantics.

\begin{lemma}
\label{lem:debit-inheritance-sound}
The debit-inheritance cost model $(\Downarrow^{\mathrm{DI}}, \Phi^{\mathrm{D}}, \stripdot^{\mathrm{D}})$ is sound.
\end{lemma}

The $\sqsubseteq^{\mathrm{DI}}$ relation differs from $\sqsubseteq^{\mathrm{D}}$
only in that the new debit is $k_1$ instead of $k$ in the [lazy] rule.

\begin{lemma}
\label{lem:debit-inheritance-persistent}
The debit-inheritance cost model $((\Downarrow^{\mathrm{DI}}, \Phi^{\mathrm{D}}, \stripdot^{\mathrm{D}}), \sqsubseteq^{\mathrm{DI}})$ is uniformly persistent.
\end{lemma}

Debit inheritance is in some sense dual to credit inheritance.
While credit inheritance passes on excess credits from a thunk to one of its children,
debit inheritance passes on debits from a child to its parent.
Aside from this similarity, the two techniques are quite different though.
For example, debit inheritance allows a parent to inherit the debits from multiple children,
while credit inheritance only allows passing on credits to a single child.

\section{Related Work}

We compare to the main related works by considering several
aspects of reasoning about persistent amortised complexity.

\paragraph*{Credits and Debits}

Okasaki~\cite{Okasaki:purefun} claims that debits are necessary for
reasoning about persistent usage, and that using credits is unsound.
His first argument centers on the lack of linearity in a persistent setting:
the persistent usage of a data structure could cause credits to be spent more than once, which would break soundness. 
But debits do not suffer from this issue: ``although savings can only be spent once, it does no harm to pay off debt more than once''~\cite[page~59]{Okasaki:purefun}.

However, this argument ignores an important aspect of thunks: they are only evaluated once.
This guarantees that the resources they hold are used linearly, even if there are many
references to the thunk itself.
Pilkiewicz and Pottier~\cite{Pilkiewicz:monotonic} and M{\'e}vel et al.~\cite{Mevel:timecredits}
exploit this property to implement a model of thunks that holds credits in the lazy state.
They show that this model does not duplicate credits and is monotonic.
Yet, these works continue to use debits in their interface of thunks.

Reasoning using credits has several advantages in practice.
For example, in a semantics where evaluating a thunk can become cheaper over time,
we can force the thunk as soon as we have enough credits for the cost at that time.
In contrast, a debit-based interface only allows forcing after assigning as many credits
as were necessary at the time of the thunk's creation.

More importantly though, the cost of evaluating a thunk may change over time.
This is most evident in the Bankers Queue, which requires the topmost
thunk of the stream to pass all its credits to the next thunk.
Pottier et al.~\cite{Pottier:thunks} model this using their Thunk-Consequence rule,
which allows them to add `save` calls to an existing thunk.
As this increases the execution cost of the thunk,
they need to increase its amount of debit retroactively.

In our work, the same can be achieved more easily by using credit inheritance,
which dynamically passes on all remaining credits after covering the execution cost.
This design is only possible because we track the
execution cost and the amount of credits separately.

\paragraph*{Reasoning about the Future}

In the absence of side-effects, it does not matter when a thunk is evaluated
(shown formally by Hackett and Hutton~\cite{Hackett:clairvoyant}).
As such, Okasaki's reasoning style may assume that thunks are evaluated immediately upon creation.
However, to obtain the right time complexity of the program,
we may not charge it the cost of evaluating the thunk upon creation.
Instead, Okasaki uses debits as a ``layaway plan'',
where the program may access the result of a thunk
only once all debits on the thunk have been paid off.

Perhaps less appreciated is that this restriction applies not just to the result of the thunk,
but also to all side-effects that are performed during the evaluation of the thunk.
In particular, if the thunk performs a side-effect, the program may not register
that the side-effect occurred until all debits have been paid off.

Crucially, saving credits on another thunk is a side-effect!
Thus, Okasaki can not generally create debits for paying off thunks,
or he would run into the soundness issue described in \cref{ex:tracking-both}.
While he does create debits for paying off thunks in his debit passing style,
he restricts it to those thunks that are guaranteed to be forced after the enclosing thunk~\cite[page~174]{Okasaki:purefun}.
As he notes, this restriction is similar to that of the debit-inheritance semantics.
The restriction is an important difference from Danielsson's version of debit passing style~\cite{Danielsson:amortisation},
which allows saving credits on any thunk. Danielsson's version is sound because
the side-effect is only executed when the thunk is forced; in contrast, Okasaki's version executes
the side-effect immediately and thus has to be more restrictive.

For some data structures, like the Bankers Queue, it is necessary to pay off a thunk that will be created
when evaluating another thunk. Okasaki's method makes it possible to save credits on the inner thunk directly.
In contrast, Danielsson~\cite{Danielsson:amortisation} assumes that thunks
are evaluated on forcing and cannot do so. Instead, he proposes deep payment as an alternative.
Deep payment injects a `save' operation into the computation of the enclosing thunk to pay
off the debits of the inner thunk once the enclosing thunk is forced.
Pottier et al.~\cite{Pottier:thunks} generalise this principle to their `Thunk-Consequence' rule,
which allows arbitrary implications to be injected.

\paragraph*{Mutable References in Thunks}

The language we consider in this paper is purely functional and does not include mutable references.
Unfortunately, almost all the results in this paper break if mutable references are included.
For example, persistence requires that the evaluation of a thunk is deterministic,
so that we do not have to revert to the lazy state when returning to a previous version.
But thunk evaluation is not guaranteed to be deterministic if the thunks can read from mutable references.
Similarly, the cost of evaluating a thunk (and thus its number of debits) may depend on the state of mutable references.
If we determine the cost of forcing a thunk given the state of the heap at the time of its creation,
then this cost may change if mutable references are written to.
Thus it would be unsound to assign a fixed number of debits to a thunk at creation time
and force it once enough debits have been assigned; after all, the cost of forcing the thunk may have gone up.
This restriction does not impact our ability to model Okasaki's data structures,
which can all be implemented without mutable references.

Pottier et al.~\cite{Pottier:thunks} avoid the requirement for deterministic execution
by specifying the result of a thunk only up to an invariant.
This makes it possible to run a non-deterministic computation in a thunk,
as long as the result satisfies the invariant.
As their work is embedded in Iris, it can interface with code that includes references and even concurrency.

\paragraph*{Persistence and Monotonicity}

Our characterisation of persistence is an instance of heap monotonicity~\cite{Fahndrich:monotonic}.
It is inspired by the work of Pilkiewicz and Pottier~\cite{Pilkiewicz:monotonic},
who show that thunks are monotonic.
In Danielsson's work~\cite{Danielsson:amortisation}, the persistence of thunks is not explicitly stated,
but follows from his type preservation result. Type preservation is explicitly related to heap
monotonicity in proofs that use logical relations~\cite{Ahmed:semantics}.
M{\'e}vel et al.~\cite{Mevel:timecredits} and Pottier et al.~\cite{Pottier:thunks}
show persistence using the persistent modality of Iris. 

There are several variants of persistence.
Confluent persistence~\cite{Driscoll:catenation} allows changes to different versions of a data structure at the same time.
For example, this makes it possible to concatenate a confluently persistent queue to a previous version of itself.
Our notion of persistence does not cover this use case,
but we relate it in the full version of this paper.
Semi-persistence allows changes to previous versions only if the current version is abandoned.
This is particularly useful for backtracking algorithms; see Allain et al.~\cite{Allain:snapshottable} for an overview.

\paragraph*{Linearity}

There is a long history of using type systems for amortised analysis~\cite{Hofmann:amortised,Hoffmann:multivariate,Simoes:lazy,Grodin:amortized};
see Hoffmann and Jost~\cite{Hoffmann:retrospective} for an overview.
Common to these approaches is the use of a linear type system to control the usage of resources.
In particular, amortised analysis is only correct if credits are not duplicated, which is guaranteed by linearity.
This makes these approaches unsuitable for amortised analysis of persistent usage.
However, they can still be used for amortised analysis of sequential usage
and worst-case analysis of persistent usage.

Jost et al.~\cite{Jost:lazy} extend amortised resource analysis to thunks,
following Danielsson's approach~\cite{Danielsson:amortisation}.
Madhavan et al.~\cite{Madhavan:contract} design a tool for estimating
the cost of functional programs with thunks. They apply their tool
to Okasaki's real-time queue and deque,
but do not propose a general reasoning principle.

The credit-based approach to persistent amortised analysis is connected to linearity
by storing credits on thunks. Thunks are only evaluated once, even if they are shared,
which guarantees that their content in the lazy state is used linearly~\cite{Pilkiewicz:monotonic,Mevel:timecredits,Lorenzen:folazy}.

\section{Conclusion}

We have presented an operational account of amortised analysis in a persistent setting.
Our characterisation of sound and persistent cost models allows us to precisely distinguish
between approaches that are persistent and those that are not.
We show that amortised analysis can be performed using credits alone, without the need for debits.
We also provide an operational semantics for Okasaki's use of debits and show that it is sound and persistent.

\paragraph*{Which Reasoning Principle is Best?}

While our work describes how lazy data structures can be analysed,
it leaves open how they should be analysed.
The credit-based analysis is quite close to the actual
implementation on a machine and allows us to provide an invariant that
describes the state of the data structure at any point in time.
However, the credit-based analysis is somewhat more involved than the debit-based
analysis, which can side-step many operational details to arrive at a shorter invariant.
Personally, we find reasoning using credits more intuitive,
but we have come to appreciate the brevity of debit-based reasoning as well.
Furthermore, while our work concerns a strict language where thunks are only used sparsely,
a lazy language where thunks are ubiquitous may instead benefit from a reasoning style
based on demand~\cite{Hackett:clairvoyant,Foner:strictcheck,Li:garden,Xia:story}.

\paragraph*{Beyond Thunks}

There could be a more general interface for persistent mutation in functional programming languages.
As we have seen, thunks are sound because updates do not change their content semantically,
and thunks are persistent because updates can only reduce the costs of future operations.
The first property is shared by quotient types, where an update that exchanges equal representatives
of the quotient does not break referential transparency~\cite{Selsam:sealing}.
In fact, thunks can be seen as an instance of this principle~\cite{Lorenzen:folazy}.
However, we additionally need to ensure monotonicity of updates.
For example, it would not be persistent to revert a thunk back to its lazy state after it has been memoised.
We are interested in exploring a more general interface for monotonic updates,
like LVars~\cite{Kuper:lvars} or stable references~\cite{Ponsonnet:catenable,Kaplan:catenable}.

\bibliography{camera-ready}

\appendix

\newpage

\section{Proofs}

\subsection{Soundness}

\paragraph*{\cref{lem:bankers-sound}}
\begin{proof}
Induction over the derivation of $\Gamma : e \Downarrow^{\mathrm{B}}_n \Delta : v$.
We only have to consider the (save) and (spend) rules, which are new to this cost model.
In the real cost model, $\save{m}{a}$ is a no-op and $\spend{m}{a}{e}$ just evaluates $e$.
For (save), we have:
\begin{align*}
  \Phi^{\mathrm{B}}(\Gamma, a \mapsto_{n+m} hv) - \Phi^{\mathrm{B}}(\Gamma, a \mapsto_n hv)
  &= (\Phi^{\mathrm{B}}(\Gamma) + n + m) - (\Phi^{\mathrm{B}}(\Gamma) + n) \\
  &= m
\end{align*}
For (spend), we apply the lemma inductively to the premise.
Assume the real semantics takes $k$ steps and the Bankers semantics $k'$ steps on the premise.
Then we have $k + \Phi^{\mathrm{B}}(\Delta) - \Phi^{\mathrm{B}}(\Gamma, a \mapsto_{n} hv) \leq k'$.
We have:
\begin{align*}
  k + \Phi^{\mathrm{B}}(\Delta) - \Phi^{\mathrm{B}}(\Gamma, a \mapsto_{n+m} hv)
  &= k + \Phi^{\mathrm{B}}(\Delta) - (\Phi^{\mathrm{B}}(\Gamma) + n + m) \\
  &= k + \Phi^{\mathrm{B}}(\Delta) - \Phi^{\mathrm{B}}(\Gamma, a \mapsto_{n} hv) - m \\
  &\leq k' - m \\
  &\leq \max(k' - m, 0) 
\end{align*}
\end{proof}

\paragraph*{\cref{lem:credit-passing-sound}}
\begin{proof}
Induction over the derivation of $\Gamma : e \Downarrow^{\mathrm{C}}_n \Delta : v$.
We only consider the new rules. Again, $\save{m}{a}$ is a no-op in the real cost model.

\begin{itemize}
\bitem{(lazy)} We have
$\strip{\Gamma, a \mapsto_0 \lazy_F\; v}^{\mathrm{C}} = \strip{\Gamma}^{\mathrm{C}}, a \mapsto \lazy_F\; v$
and
\begin{align*}
  \Phi^{\mathrm{C}}(\Gamma, a \mapsto_0 \lazy_F\; v) - \Phi^{\mathrm{C}}(\Gamma)
  &= (\Phi^{\mathrm{C}}(\Gamma) + 0) - \Phi^{\mathrm{C}}(\Gamma) \\
  &= 0
\end{align*}
\bitem{(waste)} Trivial, since $\save{m}{a}$ is a no-op in this rule.
\bitem{(save)}
  We have $\strip{\Gamma, a \mapsto_{n} \lazy_F\; v}^{\mathrm{C}} = \strip{\Gamma}^{\mathrm{C}}, a \mapsto \lazy_F\; v = \strip{\Gamma, a \mapsto_{n+m} \lazy_F\; v}^{\mathrm{C}}$
  and
  \begin{align*}
  \Phi^{\mathrm{C}}(\Gamma, a \mapsto_{n + m} \lazy_F\;v) - \Phi^{\mathrm{C}}(\Gamma, a \mapsto_n \lazy_F\;v)
  &= \Phi^{\mathrm{C}}(\Gamma) + n + m - \Phi^{\mathrm{C}}(\Gamma) - n \\
  &= m
  \end{align*}
\bitem{(force)} By the induction hypothesis, we have
  $\strip{\Gamma}^{\mathrm{C}} : F\, v \Downarrow^{\mathrm{R}}_{k'} \strip{\Delta}^{\mathrm{C}} : w$ and $k' + \Phi^{\mathrm{C}}(\Delta) - \Phi^{\mathrm{C}}(\Gamma) \leq k$.
  By the (force) rule, we have
  $\strip{\Gamma, a \mapsto_n \lazy_F\; v}^{\mathrm{C}} : \force a \Downarrow^{\mathrm{R}}_{k'} \strip{\Delta, a \mapsto \memo\; w}^{\mathrm{C}} : w$.
  The second condition holds as:
  \begin{align*}
  k' + \Phi^{\mathrm{C}}(\Delta, a \mapsto \memo\; w) - \Phi^{\mathrm{C}}(\Gamma, a \mapsto_n \lazy_F\; v)
  &= k' + \Phi^{\mathrm{C}}(\Delta) - \Phi^{\mathrm{C}}(\Gamma) - n \\
  &\leq k - n \\
  &\leq n - n \\
  &= 0
  \end{align*}
\end{itemize}
\end{proof}

\begin{lemma}
\label{lem:credit-inheritance-save-sound}
If $\Gamma : \save{m}{a} \Downarrow^{\mathrm{CI}}_m \Delta : v$,
then $\strip{\Gamma}^{\mathrm{CI}} = \strip{\Delta}^{\mathrm{CI}}$,
$\strip{\Gamma}^{\mathrm{CI}} : \save{m}{a} \Downarrow^{\mathrm{R}}_0 \strip{\Delta}^{\mathrm{CI}} : a$
and $\Phi^{\mathrm{CI}}(\Delta) - \Phi^{\mathrm{CI}}(\Gamma) \leq m$.
\end{lemma}
\begin{proof}
By induction over the derivation of $\Gamma : \save{m}{a} \Downarrow^{\mathrm{CI}}_m \Delta : v$.
\begin{itemize}
\bitem{(waste)} Trivial, since $\Phi^{\mathrm{CI}}(\Gamma) - \Phi^{\mathrm{CI}}(\Gamma) = 0 \leq m$.
\bitem{(save)}
  We have $\strip{\Gamma, a \mapsto_{n} \lazy_F\; v}^{\mathrm{CI}} = \strip{\Gamma}^{\mathrm{CI}}, a \mapsto \lazy_F\; v = \strip{\Gamma, a \mapsto_{n+m} \lazy_F\; v}^{\mathrm{CI}}$
  and
  \begin{align*}
  \Phi^{\mathrm{CI}}(\Gamma, a \mapsto_{n + m} \lazy_F\;v) - \Phi^{\mathrm{CI}}(\Gamma, a \mapsto_n \lazy_F\;v)
  &= \Phi^{\mathrm{CI}}(\Gamma) + n + m - \Phi^{\mathrm{CI}}(\Gamma) - n \\
  &= m
  \end{align*}
\bitem{(inherit)}
  By the induction hypothesis, we have $\strip{\Gamma}^{\mathrm{CI}} = \strip{\Delta}^{\mathrm{CI}}$,
  $\strip{\Gamma}^{\mathrm{CI}} : \save{m}{h} \Downarrow^{\mathrm{R}}_0 \strip{\Delta}^{\mathrm{CI}} : h$
  and $\Phi^{\mathrm{CI}}(\Delta) - \Phi^{\mathrm{CI}}(\Gamma) \leq m$.
  We have
\begin{align*}
  \strip{\Gamma, a \mapsto_h \memo\; w}^{\mathrm{CI}}
  &= \strip{\Gamma}^{\mathrm{CI}}, a \mapsto \memo\; w \\
  &= \strip{\Delta}^{\mathrm{CI}}, a \mapsto \memo\; w \\
  &= \strip{\Delta, a \mapsto_h \memo\; w}^{\mathrm{CI}}
\end{align*}
and thus $\strip{\Gamma, a \mapsto_h \memo\; w}^{\mathrm{CI}} : \save{m}{a} \Downarrow^{\mathrm{R}}_0 \strip{\Delta, a \mapsto_h \memo\; w}^{\mathrm{CI}} : a$
and
\begin{align*}
  \Phi^{\mathrm{CI}}(\Delta, a \mapsto_h \memo\; w) - \Phi^{\mathrm{CI}}(\Gamma, a \mapsto_h \memo\; w)
  &= \Phi^{\mathrm{CI}}(\Delta) - \Phi^{\mathrm{CI}}(\Gamma) \\
  &\leq m
\end{align*}
\end{itemize}
\end{proof}

\paragraph*{\cref{lem:credit-inheritance-sound}}
\begin{proof}
By induction over the derivation of $\Gamma : e \Downarrow^{\mathrm{CI}}_n \Delta : v$.
We only consider the new rules. Again, $\save{m}{a}$ and $\pass{h}$ are no-ops in the real cost model.
\begin{itemize}
\bitem{(pass)} By \cref{lem:credit-inheritance-save-sound}, we have
  $\strip{\Gamma}^{\mathrm{CI}} = \strip{\Delta}^{\mathrm{CI}}$,
  $\strip{\Gamma}^{\mathrm{CI}} : \save{m}{h} \Downarrow^{\mathrm{R}}_0 \strip{\Delta}^{\mathrm{CI}} : h$
  and $\Phi^{\mathrm{CI}}(\Delta) - \Phi^{\mathrm{CI}}(\Gamma) \leq m$.
  Since $\pass{h}$ is a no-op in the real cost model, this concludes the proof.
\bitem{(force)} By the induction hypothesis, we have
  $\strip{\Gamma}^{\mathrm{C}} : F\, v \Downarrow^{\mathrm{R}}_{k'} \strip{\Delta}^{\mathrm{C}} : w$ and $k' + \Phi^{\mathrm{C}}(\Delta) - \Phi^{\mathrm{C}}(\Gamma) \leq k$.
  By the (force) rule, we have
  $\strip{\Gamma, a \mapsto_n \lazy_F\; v}^{\mathrm{C}} : \force a \Downarrow^{\mathrm{R}}_{k'} \strip{\Delta, a \mapsto_h \memo\; w}^{\mathrm{C}} : w$.
  The second condition holds as:
  \begin{align*}
  k' + \Phi^{\mathrm{C}}(\Delta, a \mapsto_h \memo\; w) - \Phi^{\mathrm{C}}(\Gamma, a \mapsto_n \lazy_F\; v)
  &= k' + \Phi^{\mathrm{C}}(\Delta) - \Phi^{\mathrm{C}}(\Gamma) - n \\
  &\leq k - n \\
  &= n - n \\
  &= 0
  \end{align*}
\bitem{(inherit)} By \cref{lem:credit-inheritance-save-sound}.
\end{itemize}
\end{proof}

\begin{remark}
For the debit cost model, we ensure the following invariant on heaps:
For every heap $\Gamma, a \mapsto_k \memo^{\Delta}_{F\,v}\; w$, we require that
$\strip{\Gamma\setminus\Delta}^{\mathrm{D}} : F\, v \Downarrow^{\mathrm{R}}_{k^\mathrm{R}} \strip{\Gamma}^{\mathrm{D}} : w$
and that $k \leq k^{\mathrm{R}}$.
\end{remark}

\begin{lemma}
\label{lem:debit-cost}
If $\Gamma : e \Downarrow^{\mathrm{D}}_{k,k'} \Delta : v$,
and $\strip{\Gamma}^{\mathrm{D}} : e \Downarrow^{\mathrm{R}}_{k^{\mathrm{R}}} \strip{\Delta}^{\mathrm{D}} : v$,
then $k \leq k^{\mathrm{R}}$.
\end{lemma}
\begin{proof}
By induction over the derivation of $\Gamma : e \Downarrow^{\mathrm{D}}_{k,k'} \Delta : v$.
The property hold for all rules.
\end{proof}

\paragraph*{\cref{lem:debit-sound}}
\begin{proof}
We first show that both semantics produce the same final heap and value.
By induction over the derivation of $\Gamma : e \Downarrow^{\mathrm{D}}_{k,k'} \Delta : v$.
We only consider the new rules.
\begin{itemize}
\bitem{(lazy)} We have
\begin{align*}
  \strip{\Delta, a \mapsto_k \memo^{\Delta\setminus\Gamma}_{F\,v}\; w}^{\mathrm{D}}
  &= \strip{\Delta \setminus (\Delta \setminus \Gamma)}^{\mathrm{D}}, a \mapsto \lazy_F\; v \\
  &= \strip{\Gamma}^{\mathrm{D}}, a \mapsto \lazy_F\; v
\end{align*}
By the induction hypothesis, we have
$\strip{\Gamma}^{\mathrm{D}} : e \Downarrow^{\mathrm{R}}_{k} \strip{\Delta}^{\mathrm{D}} : v$,
and thus our invariant holds.
\bitem{(force)} We have
\begin{align*}
  \strip{\Gamma, a \mapsto_n \memo^{\Delta}_{F\,v}\; w}^{\mathrm{D}}
  &= \strip{\Gamma \setminus \Delta}^{\mathrm{D}}, a \mapsto \lazy_F\; v
\end{align*}
and
\begin{align*}
  \strip{\Gamma, a \mapsto \memo\; w}^{\mathrm{D}}
  &= \strip{\Gamma}^{\mathrm{D}}, a \mapsto \memo\; w
\end{align*}
By our invariant, we have
$\strip{\Gamma\setminus\Delta}^{\mathrm{D}} : F\, v \Downarrow^{\mathrm{R}}_k \strip{\Gamma}^{\mathrm{D}} : w$
We can thus apply the (force) rule in the real cost model to obtain the claim.
\bitem{(save), (waste)} Trivial, since $\strip{\Gamma, a \mapsto_n \memo^{\Delta}_{F\,v}\; w}^{\mathrm{D}} = \strip{\Gamma}^{\mathrm{D}}, a \mapsto \memo\; w = \strip{\Gamma, a \mapsto_{n - m} \memo^{\Delta}_{F\,v}\; w}^{\mathrm{D}}$.
\end{itemize}

Now we show the cost condition.
By induction over the derivation of $\Gamma : e \Downarrow^{\mathrm{D}}_{k,k'} \Delta : v$.
We only consider the new rules.
\begin{itemize}
\bitem{(lazy)} Let $k^{\mathrm{R}}$ be the time taken for the thunk in the real semantics.
By the induction hypothesis, we have
$k^{\mathrm{R}} + \Phi^{\mathrm{D}}(\Delta) - \Phi^{\mathrm{D}}(\Gamma) \leq k + k'$.
We have
\begin{align*}
  \Phi^{\mathrm{D}}(\Delta, a \mapsto_k \memo^{\Delta}_{F\,v}\; w) - \Phi^{\mathrm{D}}(\Gamma)
  &= (\Phi^{\mathrm{D}}(\Delta) + k^{\mathrm{R}} - k) - \Phi^{\mathrm{D}}(\Gamma) \\
  &\leq k^{\mathrm{R}} - k + k + k' - k^{\mathrm{R}} \\
  &\leq k'
\end{align*}
By \cref{lem:debit-cost}, we have $k \leq k^{\mathrm{R}}$, which ensures that $\Phi^{\mathrm{D}}$ remains non-negative.
\bitem{(force)} Let $k^{\mathrm{R}}$ be the time taken for the thunk in the real semantics.
\begin{align*}
  &k^{\mathrm{R}} + \Phi^{\mathrm{D}}(\Gamma, a \mapsto \memo\; w) - \Phi^{\mathrm{D}}(\Gamma, a \mapsto_n \memo^{\Delta}_{F\,v}\; w) \\
  &= k^{\mathrm{R}} + \Phi^{\mathrm{D}}(\Gamma) - \Phi^{\mathrm{D}}(\Gamma) - k^{\mathrm{R}} + n \\
  &\leq n \\
  &\leq 0
\end{align*}
\bitem{(save)} We have
\begin{align*}
  &\Phi^{\mathrm{D}}(\Gamma, a \mapsto_{n - m} \memo^{\Delta}_{F\,v}\; w) - \Phi^{\mathrm{D}}(\Gamma, a \mapsto_n \memo^{\Delta}_{F\,v}\; w) \\
  &= (\Phi^{\mathrm{D}}(\Gamma) + k^{\mathrm{R}} - (n - m)) - (\Phi^{\mathrm{D}}(\Gamma) + k^{\mathrm{R}} - n) \\
  &= m - n + n = m
\end{align*}
\bitem{(waste)} Trivial, since it doesn't change the heap and $0 \leq m$.
\end{itemize}
\end{proof}

\begin{remark}
For the debit inheritance cost model, we relax the invariant that $k \leq k^{\mathrm{R}}$
for every memoised thunk and include on the right hand side the costs of the child thunks as well.
\end{remark}

\begin{lemma}
\label{lem:debit-inheritance-cost}
If $\Gamma : e \Downarrow^{\mathrm{D}}_{k,k'} \Delta : v$,
$\strip{\Gamma}^{\mathrm{D}} : e \Downarrow^{\mathrm{R}}_{k^{\mathrm{R}}} \strip{\Delta}^{\mathrm{D}} : v$,
and $k_1^{\mathrm{R}}$, \ldots, $k_n^{\mathrm{R}}$ are the real costs of the child thunks,
then $k \leq k^{\mathrm{R}} + k_1^{\mathrm{R}} + \ldots + k_n^{\mathrm{R}}$.
\end{lemma}
\begin{proof}
By induction over the derivation of $\Gamma : e \Downarrow^{\mathrm{D}}_{k,k'} \Delta : v$.
We have $k \leq k^{\mathrm{R}}$ in all rules except for the new (lazy) rule with inheritance.
In that case, we have to charge for the time taken to evaluate the thunk as well.
\end{proof}

\paragraph*{\cref{lem:debit-inheritance-sound}}
\begin{proof}
By induction over the derivation of $\Gamma : e \Downarrow^{\mathrm{DI}}_{k,k'} \Delta : v$.
We only have to consider the new credit annotation on the (lazy) rule.
\begin{itemize}
\bitem{(lazy)} Let $k^{\mathrm{R}}$ be the time taken for the thunk in the real semantics.
By the induction hypothesis, we have
$k^{\mathrm{R}} + \Phi^{\mathrm{D}}(\Delta) - \Phi^{\mathrm{D}}(\Gamma) \leq k_1 + k_2 + k'$.
We have
\begin{align*}
  \Phi^{\mathrm{D}}(\Delta, a \mapsto_{k_2} \memo^{\Delta}_{F\,v}\; w) - \Phi^{\mathrm{D}}(\Gamma)
  &= (\Phi^{\mathrm{D}}(\Delta) + k^{\mathrm{R}} - k_2) - \Phi^{\mathrm{D}}(\Gamma) \\
  &\leq k^{\mathrm{R}} - k_2 + k_1 + k_2 + k' - k^{\mathrm{R}} \\
  &\leq k_1 + k'
\end{align*}
\end{itemize}
\end{proof}

\subsection{Accessibility}

To show results about persistence, it is convenient to consider only that subset
of the big-step rules that interact with the heap. Rules like (case), (split), (app) or (call)
are independent of the heap and do not affect persistence.
We thus define a version of the relation $\sqsubseteq$ that characterises
expressions that interact with the heap.
We write $\Gamma \sqsubseteq_{e,k,w} \Delta$ if $\Gamma : e \Downarrow_k \Delta : w$.
We write $\sqsubseteq_{\_,k,w}$ if different $e$ yield the same relation.

For the real cost model, the rules are as follows:

\begin{center}
  \AxiomC{\phantom{$\Gamma \sqsubseteq \Delta$}}
  \RightLabel{(refl)}
  \UnaryInfC{$\Gamma \sqsubseteq_{v,0,v} \Gamma$}
  \DisplayProof
  \quad
  \AxiomC{$\Gamma \sqsubseteq_{e_1,k,v} \Delta$}
  \AxiomC{$\Delta \sqsubseteq_{e_2[v/x],l,w} \Theta$}
  \RightLabel{(trans)}
  \BinaryInfC{$\Gamma \sqsubseteq_{\llet x = e_1 \lin e_2, k+l,w} \Theta$}
  \DisplayProof
  \quad
  \AxiomC{$\Gamma \sqsubseteq_{e, k, w} \Delta$}
  \RightLabel{(step)}
  \UnaryInfC{$\Gamma \sqsubseteq_{\_,k, w} \Delta$}
  \DisplayProof
\end{center}
\begin{center}
  \AxiomC{$a \notin \mathrm{dom}(\Gamma)$}
  \RightLabel{(alloc)}
  \UnaryInfC{$\Gamma \sqsubseteq_{\_,0,a} \Gamma, a \mapsto hv$}
  \DisplayProof
  \quad
  \AxiomC{$a \mapsto \fold v \in \Gamma$}
  \RightLabel{(unfold)}
  \UnaryInfC{$\Gamma \sqsubseteq_{\unfold a,0,v} \Gamma$}
  \DisplayProof
\end{center}
\begin{center}
  \AxiomC{$\Gamma \sqsubseteq_{F\, v, k, w} \Delta$}
  \RightLabel{(force)}
  \UnaryInfC{$(\Gamma, a \mapsto \lazy_F\; v) \sqsubseteq_{\force a, k, w} (\Delta, a \mapsto \memo\; w)$}
  \DisplayProof
\end{center}
\begin{center}
  \AxiomC{$a \mapsto \memo\, w \in \Gamma$}
  \RightLabel{(recall)}
  \UnaryInfC{$\Gamma \sqsubseteq_{\force a,0,w} \Gamma$}
  \DisplayProof
\end{center}

\begin{lemma}
\label{lem:real-accessible-equiv}
$\Gamma \sqsubseteq^{\mathrm{R}}_{e,k,w} \Delta$ iff $\Gamma : e \Downarrow^{\mathrm{R}}_k \Delta : w$.
\end{lemma}
\begin{proof}
The rules of $\sqsubseteq_{e,k,w}$ relate to the big-step semantics as follows:
\begin{itemize}
  \bitem{(refl)} (value)
  \bitem{(trans)} (let)
  \bitem{(step)} (case), (split), (app), (call)
  \bitem{(alloc)} (fold), (lazy), (memo)
  \bitem{(unfold)} (unfold)
  \bitem{(force)} (force)
  \bitem{(recall)} (recall)
\end{itemize}
\end{proof}

\begin{lemma}
\label{lem:real-accessible-def}
$\Gamma \sqsubseteq^{\mathrm{R}} \Delta$ iff there exists $e,k,w$ such that $\Gamma \sqsubseteq^{\mathrm{R}}_{e,k,w} \Delta$.
\end{lemma}
\begin{proof}
Direction ($\Rightarrow$): By induction on $\Gamma \sqsubseteq \Delta$.
We maintain that $w$ consists of nested pairs that contain all heap locations in $\Delta\setminus\Gamma$.
\begin{itemize}
  \bitem{[refl]} Immediately by the (refl) rule using $e = w = ()$.
  \bitem{[trans]} Apply the induction hypothesis on both premises,
    yielding $\Gamma \sqsubseteq_{e_1,k,v} \Delta$ and $\Delta \sqsubseteq_{e_2,l,w} \Theta$.
    Let $e = \llet x = e_1 \lin \ldots \llet y = e_2 \lin (x, y)$, where $\ldots$ 
    splits $v$ into its components. We obtain $\Gamma \sqsubseteq_{e,k',(v,w)} \Theta$
    by the (trans) and (step) rules for some $k'$.
  \bitem{[alloc]} By the (alloc) rule.
  \bitem{[force]} Apply \cref{lem:real-accessible-equiv} on the premise and then use the (force) rule.
\end{itemize}

Direction ($\Leftarrow$): By induction on $\Gamma \sqsubseteq_{e,k,w} \Delta$.
\begin{itemize}
  \bitem{(refl), (alloc)} Immediately by the [refl]/[alloc] rule.
  \bitem{(trans)} Apply the induction hypothesis on both premises and then use the [trans] rule.
  \bitem{(step)} Apply the induction hypothesis on the premise.
  \bitem{(unfold), (recall)} Use [refl].
  \bitem{(force)} Apply \cref{lem:real-accessible-equiv} on the premise
    to obtain $\Gamma : F\, v \Downarrow^{\mathrm{R}}_k \Delta : w$.
    Use the induction hypothesis to obtain $\Gamma \sqsubseteq \Delta$.
    Then use the [force] rule.
\end{itemize}
\end{proof}

\begin{lemma}
\label{lem:real-accessible}
The real cost model $((\Downarrow^{\mathrm{R}}, \Phi^{\mathrm{R}}, \stripdot^{\mathrm{R}}), \sqsubseteq^{\mathrm{R}})$ is accessible.
\end{lemma}
\begin{proof}
By \cref{lem:real-accessible-equiv} and \cref{lem:real-accessible-def}.
\end{proof}

For the Bankers cost model, we add the following two rules:

\begin{center}
  \AxiomC{\phantom{$\Gamma \sqsubseteq \Delta$}}
  \RightLabel{(save)}
  \UnaryInfC{$\Gamma, a \mapsto_n hv \sqsubseteq_{\save{m}{a}, m, a} \Gamma, a \mapsto_{n+m} hv$}
  \DisplayProof
\end{center}
\begin{center}
  \AxiomC{$\Gamma, a \mapsto_n hv \sqsubseteq_{e,k,w} \Delta$}
  \AxiomC{$n,m \geq 0$}
  \RightLabel{(spend)}
  \BinaryInfC{$\Gamma, a \mapsto_{n+m} hv \sqsubseteq_{\spend{m}{a}{e},\max(k - m,0),w} \Delta$}
  \DisplayProof
\end{center}

\begin{lemma}
\label{lem:bankers-accessible-equiv}
$\Gamma \sqsubseteq^{\mathrm{B}}_{e,k,w} \Delta$ iff $\Gamma : e \Downarrow^{\mathrm{B}}_k \Delta : w$.
\end{lemma}
\begin{proof}
The rules of $\sqsubseteq_{e,k,w}$ relate to the big-step semantics as follows:
\begin{itemize}
  \bitem{(save)} (save)
  \bitem{(spend)} (spend)
\end{itemize}
\end{proof}

\begin{lemma}
\label{lem:bankers-accessible-def}
$\Gamma \sqsubseteq^{\mathrm{B}} \Delta$ iff there exists $e,k,w$ such that $\Gamma \sqsubseteq^{\mathrm{B}}_{e,k,w} \Delta$.
\end{lemma}
\begin{proof}
Direction ($\Rightarrow$): By induction on $\Gamma \sqsubseteq \Delta$.
\begin{itemize}
  \bitem{[save]} Immediately by the (save) rule.
  \bitem{[spend]} Apply the (spend) rule using $e = \spend{m}{a}{()}$ and the [refl] rule.
\end{itemize}
Direction ($\Leftarrow$): By induction on $\Gamma \sqsubseteq^{\mathrm{B}}_{e,k,w} \Delta$.
\begin{itemize}
  \bitem{(save)} Immediately by the [save] rule.
  \bitem{(spend)} Apply the induction hypothesis on the premise, obtaining
    $\Gamma, a \mapsto_m hv \sqsubseteq \Delta$.
    Then use the [spend] rule and join the two derivations using [trans].
\end{itemize}
\end{proof}

\paragraph*{\cref{lem:bankers-accessible}}
\begin{proof}
By \cref{lem:bankers-accessible-equiv} and \cref{lem:bankers-accessible-def}.
\end{proof}

For the credit passing cost model, we add the following three rules:

\begin{center}
  \AxiomC{$hv = \lazy_F\; v$}
  \RightLabel{(save)}
  \UnaryInfC{$(\Gamma, a \mapsto_n hv) \sqsubseteq_{\save{m}{a},m, a} (\Gamma, a \mapsto_{n+m} hv)$}
  \DisplayProof
  \quad
  \AxiomC{$a \mapsto \memo\; v \in \Gamma$}
  \RightLabel{(waste)}
  \UnaryInfC{$\Gamma \sqsubseteq_{\save{m}{a},m, a} \Gamma$}
  \DisplayProof
\end{center}
\begin{center}
  \AxiomC{$\Gamma \sqsubseteq_{F\, v, k, w} \Delta$}
  \AxiomC{$k \leq n$}
  \RightLabel{(force)}
  \BinaryInfC{$(\Gamma, a \mapsto_n \lazy_F\; v) \sqsubseteq_{\force a,0,w} (\Delta, a \mapsto \memo\; w)$}
  \DisplayProof
\end{center}

\begin{lemma}
\label{lem:credit-passing-accessible-equiv}
$\Gamma \sqsubseteq^{\mathrm{C}}_{e,k,w} \Delta$ iff $\Gamma : e \Downarrow^{\mathrm{C}}_k \Delta : w$.
\end{lemma}
\begin{proof}
The rules of $\sqsubseteq_{e,k,w}$ relate to the big-step semantics as follows:
\begin{itemize}
  \bitem{(save)} (save)
  \bitem{(waste)} (waste)
  \bitem{(force)} (force)
\end{itemize}
We continue using the (recall) and (lazy) rules from the real cost model.
\end{proof}

\begin{lemma}
\label{lem:credit-passing-accessible-def}
$\Gamma \sqsubseteq^{\mathrm{C}} \Delta$ iff there exists $e,k,w$ such that $\Gamma \sqsubseteq^{\mathrm{C}}_{e,k,w} \Delta$.
\end{lemma}
\begin{proof}
Direction ($\Rightarrow$): By induction on $\Gamma \sqsubseteq \Delta$.
\begin{itemize}
  \bitem{[save]} Immediately by the (save) rule.
  \bitem{[force]} Apply \cref{lem:credit-passing-accessible-equiv} on the premise, obtaining
    $\Gamma \sqsubseteq_{F\, v, k, w} \Delta$.
    Then use the (force) rule.
\end{itemize}

Direction ($\Leftarrow$): By induction on $\Gamma \sqsubseteq^{\mathrm{C}}_{e,k,w} \Delta$.
\begin{itemize}
  \bitem{(save)} Immediately by the [save] rule.
  \bitem{(waste)} Immediately by the [refl] rule.
  \bitem{(force)} Apply \cref{lem:credit-passing-accessible-equiv} on the premise, obtaining
    $\Gamma : F\, v \Downarrow^{\mathrm{C}}_k \Delta : w$.
    Use the induction hypothesis to obtain $\Gamma \sqsubseteq \Delta$.
    Then use the [force] rule.
\end{itemize}
\end{proof}

\begin{lemma}
\label{lem:credit-passing-accessible}
The credit passing cost model $((\Downarrow^{\mathrm{C}}, \Phi^{\mathrm{C}}, \stripdot^{\mathrm{C}}), \sqsubseteq^{\mathrm{C}})$ is accessible.
\end{lemma}
\begin{proof}
By \cref{lem:credit-passing-accessible-equiv} and \cref{lem:credit-passing-accessible-def}.
\end{proof}

For the credit-inheritance cost model, we parameterise the $\sqsubseteq_{e,k,v}$ relation with a heir.

\begin{center}
  \AxiomC{$\Gamma \sqsubseteq^{\emptyset}_{\save{m}{h},m,v} \Delta$}
  \RightLabel{(pass)}
  \UnaryInfC{$\Gamma \sqsubseteq^{\{h\}}_{\pass{h},m,v} \Delta$}
  \DisplayProof
  \quad
  \AxiomC{$\Gamma \sqsubseteq^{\{h\}}_{F\, v, k, w} \Delta$}
  \AxiomC{$k = n$}
  \RightLabel{(force)}
  \BinaryInfC{$(\Gamma, a \mapsto_n \lazy_F\; v) \sqsubseteq^{\emptyset}_{\force a,0,w} (\Delta, a \mapsto_h \memo\; w)$}
  \DisplayProof
\end{center}
\begin{center}
  \AxiomC{$\Gamma \sqsubseteq^{\emptyset}_{\save{m}{h},m,v} \Delta$}
  \RightLabel{(inherit)}
  \UnaryInfC{$(\Gamma, a \mapsto_h \memo\; w) \sqsubseteq^{\emptyset}_{\save{m}{a},m,a} (\Delta, a \mapsto_h \memo\; w)$}
  \DisplayProof
\end{center}

\begin{lemma}
\label{lem:credit-inheritance-accessible-equiv}
$\Gamma \sqsubseteq^h_{e,k,w} \Delta$ iff $\Gamma : e \Downarrow^h_k \Delta : w$.
\end{lemma}
\begin{proof}
The rules of $\sqsubseteq_{e,k,w}$ relate to the big-step semantics as follows:
\begin{itemize}
  \bitem{(pass)} (pass)
  \bitem{(force)} (force)
  \bitem{(inherit)} (inherit)
\end{itemize}
\end{proof}

\begin{lemma}
\label{lem:credit-inheritance-accessible-def}
$\Gamma \sqsubseteq^{\mathrm{CI}} \Delta$ iff there exists $e,k,w$ such that $\Gamma \sqsubseteq^h_{e,k,w} \Delta$.
\end{lemma}
\begin{proof}
Direction ($\Rightarrow$): By \cref{lem:credit-passing-accessible-def}.
Direction ($\Leftarrow$):
\begin{itemize}
  \bitem{(pass)} Use the induction hypothesis on the premise to obtain $\Gamma \sqsubseteq^{\mathrm{CI}} \Delta$.
  \bitem{(force)} Apply \cref{lem:credit-inheritance-accessible-equiv} on the premise, obtaining
    $\Gamma : F\, v \Downarrow^{\{h\}}_k \Delta : w$.
    Use the induction hypothesis to obtain $\Gamma \sqsubseteq^{\mathrm{CI}} \Delta$.
    Then use the [force] rule.
  \bitem{(inherit)} Use the induction hypothesis on the premise to obtain $\Gamma \sqsubseteq^{\mathrm{CI}} \Delta$.
    Then show by induction on $\Gamma \sqsubseteq^{\mathrm{CI}} \Delta$ that
    $\Gamma, a \mapsto_h \memo\; w \sqsubseteq^{\mathrm{CI}} \Delta, a \mapsto_h \memo\; w$.
\end{itemize}
\end{proof}

\begin{lemma}
\label{lem:credit-inheritance-accessible}
The credit-inheritance cost model $((\Downarrow^{\mathrm{CI}}, \Phi^{\mathrm{CI}}, \stripdot^{\mathrm{CI}}), \sqsubseteq^{\mathrm{CI}})$ is accessible.
\end{lemma}
\begin{proof}
By \cref{lem:credit-inheritance-accessible-equiv} and \cref{lem:credit-inheritance-accessible-def}.
\end{proof}

For the debit cost model, we add the following four rules:

\begin{center}
  \AxiomC{$\Gamma \sqsubseteq_{F\, v, k, k', w} \Delta$}
  \AxiomC{$a \notin \mathrm{dom}(\Delta)$}
  \RightLabel{(lazy)}
  \BinaryInfC{$\Gamma \sqsubseteq_{\lazy_F\; v, 0, k', a} (\Delta, a \mapsto_k \memo^{\Delta\setminus\Gamma}_{F\,v}\; w)$}
  \DisplayProof
  \quad
  \AxiomC{$\memo\; w \in \Gamma$}
  \RightLabel{(waste)}
  \UnaryInfC{$\Gamma \sqsubseteq_{\save{m}{a}, 0, m, a} \Gamma$}
  \DisplayProof
\end{center}
\begin{center}
  \AxiomC{$hv = \memo^{\Delta}_{F\,v}\; w$}
  \RightLabel{(save)}
  \UnaryInfC{$(\Gamma, a \mapsto_n hv) \sqsubseteq_{\save{m}{a}, 0, m, a} (\Gamma, a \mapsto_{n - m} hv)$}
  \DisplayProof
\end{center}
\begin{center}
  \AxiomC{$n \leq 0$}
  \RightLabel{(force)}
  \UnaryInfC{$(\Gamma, a \mapsto_n \memo^{\Delta}_{F\,v}\; w) \sqsubseteq_{\force a, 0, 0, w} (\Gamma, a \mapsto \memo\; w)$}
  \DisplayProof
\end{center}

\begin{lemma}
\label{lem:debit-accessible-equiv}
$\Gamma \sqsubseteq^{\mathrm{D}}_{e,k,k',w} \Delta$ iff $\Gamma : e \Downarrow^{\mathrm{D}}_{k,k'} \Delta : w$.
\end{lemma}
\begin{proof}
The rules of $\sqsubseteq_{e,k,k',w}$ relate to the big-step semantics as follows:
\begin{itemize}
  \bitem{(lazy)} (lazy)
  \bitem{(waste)} (waste)
  \bitem{(save)} (save)
  \bitem{(force)} (force)
\end{itemize}
\end{proof}

\begin{lemma}
\label{lem:debit-accessible-def}
$\Gamma \sqsubseteq^{\mathrm{D}} \Delta$ iff there exists $e,k,k',w$ such that $\Gamma \sqsubseteq^{\mathrm{D}}_{e,k,k',w} \Delta$.
\end{lemma}
\begin{proof}
Direction ($\Rightarrow$): By induction on $\Gamma \sqsubseteq \Delta$.
\begin{itemize}
  \bitem{[lazy]} Apply \cref{lem:debit-accessible-equiv} on the premise, obtaining
    $\Gamma \sqsubseteq_{F\, v, k, k', w} \Delta$.
    Use the induction hypothesis to obtain $\Gamma \sqsubseteq \Delta$.
    Then use the (lazy) rule.
  \bitem{[save]} Immediately by the (save) rule.
  \bitem{[force]} Immediately by the (force) rule.
\end{itemize}

Direction ($\Leftarrow$): By induction on $\Gamma \sqsubseteq^{\mathrm{D}}_{e,k,k',w} \Delta$.
\begin{itemize}
  \bitem{(lazy)} Apply \cref{lem:debit-accessible-equiv} on the premise, obtaining
    $\Gamma : F\, v \Downarrow^{\mathrm{D}}_{k,k'} \Delta : w$.
    Then use the [lazy] rule.
  \bitem{(waste)} Immediately by the [refl] rule.
  \bitem{(save)} Immediately by the [save] rule.
  \bitem{(force)} Immediately by the [force] rule.
\end{itemize}
\end{proof}

\begin{lemma}
\label{lem:debit-accessible}
The debit cost model $((\Downarrow^{\mathrm{D}}, \Phi^{\mathrm{D}}, \stripdot^{\mathrm{D}}), \sqsubseteq^{\mathrm{D}})$ is accessible.
\end{lemma}
\begin{proof}
By \cref{lem:debit-accessible-equiv} and \cref{lem:debit-accessible-def}.
\end{proof}

For the debit inheritance cost model, we change the (lazy) rule:

\begin{center}
  \AxiomC{$\Gamma \sqsubseteq_{F\, v, (k_1 + k_2), k', w} \Delta$}
  \AxiomC{$a \notin \mathrm{dom}(\Delta)$}
  \RightLabel{(lazy)}
  \BinaryInfC{$\Gamma \sqsubseteq_{\lazy_F\; v, k_1, k', a} (\Delta, a \mapsto_{k_2} \memo^{\Delta\setminus\Gamma}_{F\,v}\; w)$}
  \DisplayProof
\end{center}

\begin{lemma}
\label{lem:debit-inheritance-accessible}
The debit inheritance cost model $((\Downarrow^{\mathrm{DI}}, \Phi^{\mathrm{D}}, \stripdot^{\mathrm{D}}), \sqsubseteq^{\mathrm{DI}})$ is accessible.
\end{lemma}
\begin{proof}
As in \cref{lem:debit-accessible}.
\end{proof}

\subsection{Persistence}

In the last section, we developed the $\Gamma \sqsubseteq_{e,k,w} \Delta$ relation
that holds if and only if there is a derivation $\Gamma : e \Downarrow_k \Delta : w$ of the big-step semantics,
but abstracts from the concrete expression $e$. Since the two notions are equivalent,
we prove persistence using this relation:

\begin{corollary}
\label{lem:accessibility-sqsubseteq}
An accessible cost model $((\Downarrow, \Phi, \stripdot), \sqsubseteq)$ is persistent iff
$\Gamma \sqsubseteq \Delta$ and $\Gamma \sqsubseteq_{e,n,w} \Gamma'$ implies $\Delta \sqsubseteq_{e,k,w} \Delta'$
for $k \leq n$, $\Gamma' \sqsubseteq \Delta'$.
\begin{center}
\begin{tikzcd}[row sep=large, column sep=5em, every label/.append style={scale=1.3}]
\Gamma \arrow[d, "\sqsubseteq"'] \arrow[r, "\sqsubseteq_{e,n,w}"] & \Gamma' \arrow[d, dashed, "\sqsubseteq"] \\
\Delta \arrow[r, dashed, "\sqsubseteq_{e,k,w}"] & \Delta'
\end{tikzcd}
\qquad
$(k \leq n)$
\end{center}
\end{corollary}

\begin{lemma}
\label{lem:real-weakening}
If $\Gamma \sqsubseteq^{\mathrm{R}}_{e,k,v} \Delta$ and $b \notin \mathrm{dom}(\Delta)$,
then $\Gamma, b \mapsto hv \sqsubseteq^{\mathrm{R}}_{e,k,v} \Delta, b \mapsto hv$.
\end{lemma}
\begin{proof}
By induction over the derivation of $\Gamma \sqsubseteq^{\mathrm{R}}_{e,k,v} \Delta$.
\begin{itemize}
\bitem{(refl)} Trivial
\bitem{(trans)} Use the induction hypothesis on both premises and apply (trans).
\bitem{(step)} Use the induction hypothesis on the premise and then apply (step).
\bitem{(alloc)} Trivial, as $a \neq b$ since $b \notin \mathrm{dom}(\Delta)$.
\bitem{(unfold), (recall)} Trivial
\bitem{(force)}
  Since $b \notin \mathrm{dom}(\Delta)$, $a \neq b$.
  Use the induction hypothesis on the premise and then apply (force).
\end{itemize}
\end{proof}

\paragraph*{\cref{lem:real-persistent}}
\begin{proof}
By induction over $\Gamma \sqsubseteq^{\mathrm{R}}_{e,n,w} \Gamma'$.
\begin{itemize}
\bitem{(refl)} Using (refl).
\bitem{(trans)}
  Assume that $\Gamma \sqsubseteq_{e_1,k,v} \Gamma''$ and $\Gamma'' \sqsubseteq_{e_2[v/x],l,w} \Gamma'$.
  Use the induction hypothesis on the first premise, yielding
  $\Delta \sqsubseteq_{e_1,k',v} \Delta''$
  with $k' \leq k$, $\Gamma'' \sqsubseteq \Delta''$.
  Then apply the induction hypothesis on the second premise, yielding
  $\Delta'' \sqsubseteq_{e_2[v/x],l',w} \Delta'$
  with $l' \leq l$, $\Gamma' \sqsubseteq \Delta'$.
  Then use (trans) where $k' + l' \leq k + l$ and $\Gamma' \sqsubseteq \Delta'$.
\begin{center}
\begin{tikzcd}[row sep=large, column sep=5em, every label/.append style={scale=1.3}]
\Gamma \arrow[rr, "\sqsubseteq_{\llet x = e_1 \lin e_2, k+l,w}"] & & \Gamma' \\[-2em]
\Gamma \arrow[d, "\sqsubseteq"'] \arrow[r, "\sqsubseteq_{e_1,k,v}"] & \Gamma'' \arrow[d, dashed, "\sqsubseteq"] \arrow[r, "\sqsubseteq_{e_2[v/x],l,w}"] & \Gamma' \arrow[d, dashed, "\sqsubseteq"] \\
\Delta \arrow[r, dashed, "\sqsubseteq_{e_1,k',v}"] & \Delta'' \arrow[r, dashed, "\sqsubseteq_{e_2[v/x],l',w}"] & \Delta' \\[-2em]
\Delta \arrow[rr, dashed, "\sqsubseteq_{ \llet x = e_1 \lin e_2, k'+l',w}"] & & \Delta'
\end{tikzcd}
\qquad
$(k' \leq k, l' \leq l)$
\end{center}

\bitem{(step)} Use the induction hypothesis on the premise and then apply (step).
\bitem{(alloc)} We assume w.l.o.g. that allocations never clash.
  Thus $a \not\in \mathrm{dom}(\Delta)$ and we use (alloc).
  Further, $\Gamma \sqsubseteq \Delta$ implies
  $\Gamma, a \mapsto hv \sqsubseteq \Delta, a \mapsto hv$
  by \cref{lem:real-weakening} and \cref{lem:real-accessible-def}.
\item {other cases}: By induction over $\Gamma \sqsubseteq \Delta$.
\begin{itemize}
\bitem{\_, [refl]} By assumption.
\bitem{\_, [trans]}
  Use the induction hypothesis on both premises, yielding
  $\Delta \sqsubseteq_{e,k,w} \Delta'$
  with $k \leq n$, $\Gamma' \sqsubseteq \Delta'$
  and $\Theta \sqsubseteq_{e,l,w} \Theta'$
  with $l \leq k$, $\Delta' \sqsubseteq \Theta'$.
  Then $l \leq n$ and $\Gamma' \sqsubseteq \Theta'$.
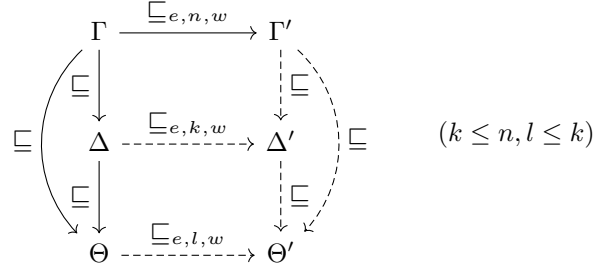
\begin{center}
\begin{tikzcd}[row sep=large, column sep=5em, every label/.append style={scale=1.3}]
\Gamma \arrow[d, "\sqsubseteq"'] \arrow[dd, bend right=45, "\sqsubseteq"'] \arrow[r, "\sqsubseteq_{e,n,w}"] & \Gamma' \arrow[d, dashed, "\sqsubseteq"] \arrow[dd, dashed, bend left=45, "\sqsubseteq"] \\
\Delta \arrow[d, "\sqsubseteq"'] \arrow[r, dashed, "\sqsubseteq_{e,k,w}"] & \Delta' \arrow[d, dashed, "\sqsubseteq"] \\
\Theta \arrow[r, dashed, "\sqsubseteq_{e,l,w}"] & \Theta'
\end{tikzcd}
\qquad
$(k \leq n, l \leq k)$
\end{center}
\bitem{\_, [alloc]}
  We assume w.l.o.g. that allocations never clash.
  Thus $a \not\in \mathrm{dom}(\Gamma')$ and we use [alloc].
  By \cref{lem:real-weakening}, $\Gamma \sqsubseteq_{e,k,v} \Gamma'$ implies
  $\Gamma, a \mapsto hv \sqsubseteq_{e,k,v} \Gamma', a \mapsto hv$.
  We have $\Gamma' \sqsubseteq \Gamma', a \mapsto hv$ by [alloc].
\bitem{(unfold), [force]} Let $\unfold a$ and $\force b$.
  We have that $a \neq b$, since $a \mapsto \fold v$
  and $b \mapsto \lazy_F\;v$.
  By the induction hypothesis,
  $\Delta \sqsubseteq_{\unfold a, 0, v} \Delta$.
  We apply \cref{lem:real-weakening} to obtain
  $\Delta, b \mapsto \memo\;w \sqsubseteq_{\unfold a, 0, v} \Delta, b \mapsto \memo\;w$.
\bitem{(recall), [force]} Let $\force a$ and $\force b$.
  We have that $a \neq b$, since $a \mapsto \memo\; w'$ and $b \mapsto \lazy_F\;v$.
  By the induction hypothesis,
  $\Delta \sqsubseteq_{\force a, 0, w'} \Delta'$.
  We apply \cref{lem:real-weakening} to obtain
  $\Delta, b \mapsto \memo\;w \sqsubseteq_{\force a, 0, w'} \Delta', b \mapsto \memo\;w$.
\bitem{(force), [force]} Let $\force a$ and $\force b$.
  If $a \neq b$, we proceed as in the (recall), [force] case.
  If $a = b$, then apply the (recall) rule to obtain
  $\Delta, b \mapsto \memo\; w \sqsubseteq_{\force a, 0, w} \Delta, b \mapsto \memo\; w$.
  We have $0 \leq k$ and $\Delta, b \mapsto \memo\; w \sqsubseteq \Delta, b \mapsto \memo\; w$ by [refl].
\begin{center}
\begin{tikzcd}[row sep=large, column sep=5em, every label/.append style={scale=1.3}]
\Gamma \arrow[d, "{\text{[force]}}"'] \arrow[r, "{\text{(force)}}"] & \Gamma' \arrow[d, dashed, "{\text{[refl]}}"] \\
\Delta \arrow[r, dashed, "{\text{(recall)}}"] & \Delta'
\end{tikzcd}
\end{center}
\end{itemize}
\end{itemize}
\end{proof}

\paragraph*{\cref{lem:credit-passing-persistent}}
\begin{proof}
We extend the proof of \cref{lem:real-persistent} for the new rules.
By induction over $\Gamma \sqsubseteq \Delta$.
\begin{itemize}
\bitem{(unfold), [save]} Let $\unfold a$ and $\save{m}{b}$.
  We have that $a \neq b$, since $a \mapsto \fold v$
  and $b \mapsto_n \lazy_F\; v$.
  Then $\Delta \sqsubseteq_{\unfold a, 0, v} \Delta$.
\bitem{(recall), [save]} Let $\force a$ and $\save{m}{b}$.
  We have that $a \neq b$, since $a \mapsto \memo\; w$
  and $b \mapsto_n \lazy_F\; v$.
  Then $\Delta \sqsubseteq_{\force a, 0, w} \Delta$.
\bitem{(save), [save]} Let $\save{m}{a}$ and $\save{m'}{b}$.
  If $a \neq b$, we proceed as in the previous cases.
  If $a = b$, then apply the [save] rule to obtain
  $\Delta, a \mapsto_{n + m'} \lazy_F\; v \sqsubseteq_{\save{m}{a}, m, a} \Delta, a \mapsto_{n + m' + m} \lazy_F\; v$.
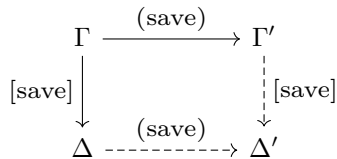
\begin{center}
\begin{tikzcd}[row sep=large, column sep=5em, every label/.append style={scale=1.3}]
\Gamma \arrow[d, "{\text{[save]}}"'] \arrow[r, "{\text{(save)}}"] & \Gamma' \arrow[d, dashed, "{\text{[save]}}"] \\
\Delta \arrow[r, dashed, "{\text{(save)}}"] & \Delta'
\end{tikzcd}
\end{center}
\bitem{(waste), [save]} Let $\save{m}{a}$ and $\save{m'}{b}$.
  We have that $a \neq b$, since $a \mapsto_n \lazy_F\; v$
  and $b \mapsto \memo\; w$.
  Then $\Delta \sqsubseteq_{\save{m}{a}, m, a} \Delta$.
\bitem{(force), [save]} Let $\force a$ and $\save{m}{b}$.
  If $a \neq b$, we proceed as in the previous cases.
  If $a = b$, then apply the [force] rule to obtain
  $\Gamma, a \mapsto_{n + m} \lazy_F\; v \sqsubseteq_{\force a, 0, a} \Delta, a \mapsto \memo\; w$
  since $k \leq n$ implies $k \leq n + m$.
\begin{center}
\begin{tikzcd}[row sep=large, column sep=5em, every label/.append style={scale=1.3}]
\Gamma \arrow[d, "{\text{[save]}}"'] \arrow[r, "{\text{(force)}}"] & \Gamma' \arrow[d, dashed, "{\text{[refl]}}"] \\
\Delta \arrow[r, dashed, "{\text{(force)}}"] & \Delta'
\end{tikzcd}
\end{center}

\bitem{(save), [force]} Let $\save{m}{a}$ and $\force b$.
  If $a \neq b$, we proceed as in the previous cases.
  If $a = b$, then apply the (waste) rule to obtain
  $\Delta, a \mapsto \memo\; w \sqsubseteq_{\save{m}{a}, m, a} \Delta, a \mapsto \memo\; w$.
\begin{center}
\begin{tikzcd}[row sep=large, column sep=5em, every label/.append style={scale=1.3}]
\Gamma \arrow[d, "{\text{[force]}}"'] \arrow[r, "{\text{(save)}}"] & \Gamma' \arrow[d, dashed, "{\text{[force]}}"] \\
\Delta \arrow[r, dashed, "{\text{(waste)}}"] & \Delta'
\end{tikzcd}
\end{center}
\bitem{(waste), [force]} Let $\save{m}{a}$ and $\force b$.
  We have that $a \neq b$, since $a \mapsto \memo\; w$
  and $b \mapsto_n \lazy_F\; v$.
  Then $\Delta \sqsubseteq_{\force b, 0, w'} \Delta'$.
\bitem{((unfold), [force]), ((recall), [force]), ((force), [force])} As in \cref{lem:real-persistent}.
\end{itemize}
\end{proof}

In the credit inheritance semantics, the number of time steps is not a deterministic function of the initial heap and expression.
Instead, if there is an heir, we can choose the amount of steps non-deterministically:
$\Downarrow^{\{h\}} : (\mathrm{Heap} \times \mathrm{Expr} \times \mathbb{N}) \rightharpoonup \mathrm{Heap} \times \mathrm{Value}$.
However, we can show that if we use more steps than necessary,
this only results in an increase in the credit assigned to the heir.

\begin{lemma}
\label{lem:credit-inheritance-pass}
Let $\Gamma : e \Downarrow^{\{h\}}_n \Delta : w$
and $\Gamma : e \Downarrow^{\{h\}}_m \Delta' : w'$ for $n \leq m$.
Then $w = w'$ and $\Delta \sqsubseteq_{\save{m - n}{h},m - n,h} \Delta'$.
\end{lemma}
\begin{proof}
We strengthen the lemma to allow the starting heap to differ:

Let $\Gamma : e \Downarrow^{\{h\}}_n \Delta : w$
and $\Gamma' : e \Downarrow^{\{h\}}_m \Delta' : w'$
for $n \leq m$ and $\Gamma \sqsubseteq_{\save{k' - k}{h},k' - k,h} \Gamma'$.
Then $w = w'$ and $\Delta \sqsubseteq_{\save{k' - k + m - n}{h}, k' - k + m - n,h} \Delta'$.

By induction over the derivations.
All cases follow directly from the induction hypothesis,
except for the (pass), (save) and (force) rules.
\begin{itemize}
\bitem{(pass)} We have:
\begin{align*}
  \Gamma &\sqsubseteq_{\save{k' - k}{h},k' - k,h} \Gamma' \\
  \Gamma' &\sqsubseteq^{\{h\}}_{\pass{h},m,h} \Delta' \\
  \Gamma &\sqsubseteq^{\{h\}}_{\pass{h},n,h} \Delta
\end{align*}
  The save rule does not change the heap except for adding credits.
  This implies that $\save{\cdot}{h}$ credits the same cell on all heaps.
  We can combine the first two facts to obtain $\Gamma \sqsubseteq_{\save{k' - k + m}{h},k' - k + m,h} \Delta'$.
  By the last fact, thus $\Delta \sqsubseteq_{\save{k' - k + m - n}{h},k' - k + m - n,h} \Delta'$.
\bitem{(save)} Let $\save{m}{a}$. Then $n = m$ and $w = w'$ by the induction hypothesis.
  If $a \neq h$, the claim follows by induction.
  Let $a = h$. We have $\Gamma \sqsubseteq_{\save{k' - k}{h},k' - k,h} \Gamma'$ by the premise.
  Again, we can combine this with the (save) rule to obtain
  $\Delta \sqsubseteq_{\save{k' - k + m}{h},k' - k + m,h} \Delta'$
\bitem{(force)} Let $\force a$. Then $n = m$ and $w = w'$ by the induction hypothesis.
  If $a \neq h$, the claim follows by induction.
  Let $a = h$. We have $\Gamma \sqsubseteq_{\save{k' - k}{h},k' - k,h} \Gamma'$ by the premise.
  Let $a \mapsto_{k} \lazy_F\; v \in \Gamma$ and $a \mapsto_{k'} \lazy_F\; v \in \Gamma'$.
  By the induction hypothesis, we get
  $\Gamma : F\, v \Downarrow^{\{h'\}}_k \Delta : w$ and
  $\Gamma' : F\, v \Downarrow^{\{h'\}}_{k'} \Delta' : w$
  and $\Delta \sqsubseteq_{\save{k' - k}{h'},k' - k,h} \Delta'$.
  In the resulting heaps $\Delta$ and $\Delta'$, we now have $h \mapsto_{h'} \memo\;w$.
  By the (inherit) rule, we get $\Delta \sqsubseteq_{\save{k' - k}{h},k' - k,h} \Delta'$.
\end{itemize}
\end{proof}

\paragraph*{\cref{lem:credit-inheritance-persistent}}
\begin{proof}
We extend the proof of \cref{lem:credit-passing-persistent} for the new rules.
\begin{itemize}
\bitem{(pass)} Follows directly from the induction hypothesis on the premise.
\bitem{(inherit)} By the induction hypothesis on the premise, we have
  $\Delta \sqsubseteq_{\save{m}{h}, m, h} \Delta'$ and $\Gamma' \sqsubseteq \Delta'$.
  We then get $\Delta, a \mapsto_h \memo\; w \sqsubseteq_{\save{m}{a}, m, a} \Delta', a \mapsto_h \memo\; w$
  by the (inherit) rule. We have $\Gamma', a \mapsto_h \memo\; w \sqsubseteq \Delta', a \mapsto_h \memo\; w$
  by \cref{lem:real-weakening}.
\bitem{other cases:} By induction over $\Gamma \sqsubseteq \Delta$.
\begin{itemize}
\bitem{(save), [force]} Let $\save{m}{a}$ and $\force b$.
  If $a \neq b$, we proceed as in the previous proof.
  If $a = b$, then $\Gamma = \Gamma_1, a \mapsto_n \lazy_F\;v$
  and $\Delta = \Delta_1, a \mapsto_h \memo\;w$.
  Apply the (inherit) rule to obtain $\Delta \sqsubseteq_{\save{m}{a}, m, a} \Delta'$.
  By [force], we have $\Gamma_1 : F\, v \Downarrow^{\{h\}}_n \Delta_1 : w$.
  By \cref{lem:credit-inheritance-pass}, we have $\Gamma_1 : F\, v \Downarrow^{\{h\}}_{n+m} \Delta_1' : w$
  for $\Delta_1 \sqsubseteq_{\save{m}{a}, m, a} \Delta_1'$.
  Apply \cref{lem:real-weakening} to obtain $\Delta_1, a \mapsto_h \memo\; w \sqsubseteq_{\save{m}{a}, m, a} \Delta_1', a \mapsto_h \memo\; w$.
\begin{center}
\begin{tikzcd}[row sep=large, column sep=5em, every label/.append style={scale=1.3}]
\Gamma \arrow[d, "{\text{[force]}}"'] \arrow[r, "{\text{(save)}}"] & \Gamma' \arrow[d, dashed, "{\text{[force]}}"] \\
\Delta \arrow[r, dashed, "{\text{(inherit)}}"] & \Delta'
\end{tikzcd}
\end{center}
\bitem{(force), [save]} Let $\force a$ and $\save{m}{b}$.
  If $a \neq b$, we proceed as in the previous proof.
  If $a = b$, then we have
  $\Gamma = \Gamma_1, a \mapsto_n \lazy_F\; v$
  and $\Gamma' = \Gamma_1', a \mapsto_h \memo\; w$
  where $\Gamma_1 : F\, v \Downarrow^{\{h\}}_n \Gamma'_1 : w$.
  Thus also $\Gamma_1 : F\, v \Downarrow^{\{h\}}_{n+m} \Gamma_1'' : w$ and
  $\Gamma_1, a \mapsto_{n+m} \lazy_F\; v \Downarrow_1 \Gamma_1'', a \mapsto_h \memo\; w$
  by the (force) rule.
  By \cref{lem:credit-inheritance-pass}, we have $\Gamma_1' \sqsubseteq_{\save{m}{a}, m, a} \Gamma_1''$.
  By \cref{lem:real-weakening}, we get
  $\Gamma_1', a \mapsto_h \memo\; w \sqsubseteq_{\save{m}{a}, m, a} \Gamma_1'', a \mapsto_h \memo\; w$.
\begin{center}
\begin{tikzcd}[row sep=large, column sep=5em, every label/.append style={scale=1.3}]
\Gamma \arrow[d, "{\text{[save]}}"'] \arrow[r, "{\text{(force)}}"] & \Gamma' \arrow[d, dashed, "{\text{[save]}}"] \\
\Delta \arrow[r, dashed, "{\text{(force)}}"] & \Delta'
\end{tikzcd}
\end{center}
\bitem{(force), [force]} As in the proof of \cref{lem:real-persistent} since the $\Downarrow^{\{h\}}$ relation
  is deterministic for fixed step count as shown by \cref{lem:credit-inheritance-pass}.
\begin{center}
\begin{tikzcd}[row sep=large, column sep=5em, every label/.append style={scale=1.3}]
\Gamma \arrow[d, "{\text{[force]}}"'] \arrow[r, "{\text{(force)}}"] & \Gamma' \arrow[d, dashed, "{\text{[refl]}}"] \\
\Delta \arrow[r, dashed, "{\text{(recall)}}"] & \Delta'
\end{tikzcd}
\end{center}
\end{itemize}
\end{itemize}
\end{proof}

\paragraph*{\cref{lem:debit-persistent}}
\begin{proof}
We extend the proof of \cref{lem:real-persistent} for the new rules.
\begin{itemize}
\bitem{(lazy)} By the induction hypothesis on the premise, we have
  $\Delta \sqsubseteq_{F\, v, k'', k', w} \Delta'$ for $k'' \leq k$ and $\Gamma' \sqsubseteq \Delta'$.
  We then get $\Delta \sqsubseteq_{\lazy_F\; v, 0, k', a} \Delta', a \mapsto_{k''} \memo^{\Delta\setminus\Gamma}_{F\,v}\; w$
  by the (lazy) rule. We have $\Gamma', a \mapsto_{k''} \memo^{\Delta\setminus\Gamma}_{F\,v}\; w \sqsubseteq \Delta', a \mapsto_{k''} \memo^{\Delta\setminus\Gamma}_{F\,v}\; w$
  by the [trans] of $\Gamma' \sqsubseteq \Delta'$ and [save].
\begin{center}
\begin{tikzcd}[row sep=large, column sep=5em, every label/.append style={scale=1.3}]
\Gamma \arrow[d, "\sqsubseteq"'] \arrow[r, "{\text{(lazy)}}"] & \Gamma' \arrow[d, dashed, "\sqsubseteq {\text{ [trans] [save]}}"] \\
\Delta \arrow[r, dashed, "{\text{(lazy)}}"] & \Delta'
\end{tikzcd}
\end{center}
\item {other cases:} By induction over $\Gamma \sqsubseteq \Delta$.
\begin{itemize}
\bitem{\_, [lazy]} By the induction hypothesis we have 
  $\Delta \sqsubseteq_{e, l, k', v} \Delta'$ for some $l \leq k$.
  Since $a \notin \mathrm{dom}(\Delta')$, this implies by \cref{lem:real-weakening} that
  $\Delta, a \mapsto_{k''} \memo^{\Delta\setminus\Gamma}_{F\,v}\; w \sqsubseteq_{e, l, k', v} \Delta', a \mapsto_{k''} \memo^{\Delta\setminus\Gamma}_{F\,v}\; w$.
  We get $\Gamma' \sqsubseteq \Delta', a \mapsto_{k''} \memo^{\Delta\setminus\Gamma}_{F\,v}\; w$ by [lazy].
\begin{center}
\begin{tikzcd}[row sep=large, column sep=5em, every label/.append style={scale=1.3}]
\Gamma \arrow[d, "{\text{[lazy]}}"'] \arrow[r, "\sqsubseteq_{e, k, k', v}"] & \Gamma' \arrow[d, dashed, "{\text{[lazy]}}"] \\
\Delta \arrow[r, dashed, "\sqsubseteq_{e, l, k', v}"] & \Delta'
\end{tikzcd}
\end{center}

\bitem{((unfold), [save]), ((recall), [save]), ((waste), [save]), ((waste), [force]), ((unfold), [force]), ((recall), [force])}
  In all of these rules, the addresses they act on are different.
  This allows us to permute the rules as in the previous proofs.
\bitem{(save), [save]} Let $\save{m}{a}$ and $\save{m'}{b}$.
  If $a \neq b$, we proceed as in the previous cases.
  If $a = b$, then apply the [save] rule to obtain
  $\Delta, a \mapsto_{n - m'} \lazy_F\; v \sqsubseteq_{\save{m}{a}, 0, m, a} \Delta, a \mapsto_{n - m' - m} \lazy_F\; v$.
\begin{center}
\begin{tikzcd}[row sep=large, column sep=5em, every label/.append style={scale=1.3}]
\Gamma \arrow[d, "{\text{[save]}}"'] \arrow[r, "{\text{(save)}}"] & \Gamma' \arrow[d, dashed, "{\text{[save]}}"] \\
\Delta \arrow[r, dashed, "{\text{(save)}}"] & \Delta'
\end{tikzcd}
\end{center}
\bitem{(force), [save]} Let $\force a$ and $\save{m}{b}$.
  If $a \neq b$, we proceed as in the previous cases.
  If $a = b$, then apply the (force) rule to obtain
  $\Gamma, a \mapsto_{n - m} \memo^{\Delta}_{F\,v}\; w \sqsubseteq_{\force a, 0, 0, a} \Delta, a \mapsto \memo\; w$
  since $n \leq 0$ implies $n - m \leq 0$.
\begin{center}
\begin{tikzcd}[row sep=large, column sep=5em, every label/.append style={scale=1.3}]
\Gamma \arrow[d, "{\text{[save]}}"'] \arrow[r, "{\text{(force)}}"] & \Gamma' \arrow[d, dashed, "{\text{[refl]}}"] \\
\Delta \arrow[r, dashed, "{\text{(force)}}"] & \Delta'
\end{tikzcd}
\end{center}

\bitem{(save), [force]} Let $\save{m}{a}$ and $\force b$.
  If $a \neq b$, we proceed as in the previous cases.
  If $a = b$, then apply the (waste) rule to obtain
  $\Delta, a \mapsto \memo\; w \sqsubseteq_{\save{m}{a}, 0, m, a} \Delta, a \mapsto \memo\; w$.
\begin{center}
\begin{tikzcd}[row sep=large, column sep=5em, every label/.append style={scale=1.3}]
\Gamma \arrow[d, "{\text{[force]}}"'] \arrow[r, "{\text{(save)}}"] & \Gamma' \arrow[d, dashed, "{\text{[force]}}"] \\
\Delta \arrow[r, dashed, "{\text{(waste)}}"] & \Delta'
\end{tikzcd}
\end{center}
\bitem{(force), [force]} Let $\force a$ and $\force b$.
  If $a \neq b$, we proceed as in the previous cases.
  If $a = b$, then apply the (recall) rule to obtain
  $\Delta, a \mapsto \memo\; w \sqsubseteq_{\force a, 0, 0, a} \Delta, a \mapsto \memo\; w$.
  We have $\Delta, a \mapsto \memo\; w \sqsubseteq \Delta, a \mapsto \memo\; w$ by [refl].
\begin{center}
\begin{tikzcd}[row sep=large, column sep=5em, every label/.append style={scale=1.3}]
\Gamma \arrow[d, "{\text{[force]}}"'] \arrow[r, "{\text{(force)}}"] & \Gamma' \arrow[d, dashed, "{\text{[refl]}}"] \\
\Delta \arrow[r, dashed, "{\text{(recall)}}"] & \Delta'
\end{tikzcd}
\end{center}
\end{itemize}
\end{itemize}
\end{proof}

\paragraph*{\cref{lem:debit-inheritance-persistent}}
\begin{proof}
We extend the proof of \cref{lem:debit-persistent} for the new (lazy) rule.
\begin{itemize}
\bitem{(lazy)} By the induction hypothesis on the premise, we have
  $\Delta \sqsubseteq_{F\, v, (k_1 + k_2), k', w} \Delta'$ for $\Gamma' \sqsubseteq \Delta'$.
  We then get $\Delta \sqsubseteq_{\lazy_F\; v, k_1, k', a} \Delta', a \mapsto_{k_2} \memo^{\Delta\setminus\Gamma}_{F\,v}\; w$
  by the (lazy) rule. We have $\Gamma', a \mapsto_{k_2} \memo^{\Delta\setminus\Gamma}_{F\,v}\; w \sqsubseteq \Delta', a \mapsto_{k_2} \memo^{\Delta\setminus\Gamma}_{F\,v}\; w$
  by $\Gamma' \sqsubseteq \Delta'$ and \cref{lem:real-weakening}.
\begin{center}
\begin{tikzcd}[row sep=large, column sep=5em, every label/.append style={scale=1.3}]
\Gamma \arrow[d, "\sqsubseteq"'] \arrow[r, "{\text{(lazy)}}"] & \Gamma' \arrow[d, dashed, "\sqsubseteq"] \\
\Delta \arrow[r, dashed, "{\text{(lazy)}}"] & \Delta'
\end{tikzcd}
\end{center}
\bitem{\_, [lazy]} By the induction hypothesis we have 
  $\Delta \sqsubseteq_{e, k, k', v} \Delta'$.
  Since $a \notin \mathrm{dom}(\Delta')$, this implies by \cref{lem:real-weakening} that
  $\Delta, a \mapsto_{k_2} \memo^{\Delta\setminus\Gamma}_{F\,v}\; w \sqsubseteq_{e, k, k', v} \Delta', a \mapsto_{k_2} \memo^{\Delta\setminus\Gamma}_{F\,v}\; w$.
  We get $\Gamma' \sqsubseteq \Delta', a \mapsto_{k_2} \memo^{\Delta\setminus\Gamma}_{F\,v}\; w$ by [lazy].
\begin{center}
\begin{tikzcd}[row sep=large, column sep=5em, every label/.append style={scale=1.3}]
\Gamma \arrow[d, "{\text{[lazy]}}"'] \arrow[r, "\sqsubseteq_{e, k, k', v}"] & \Gamma' \arrow[d, dashed, "{\text{[lazy]}}"] \\
\Delta \arrow[r, dashed, "\sqsubseteq_{e, k, k', v}"] & \Delta'
\end{tikzcd}
\end{center}
\end{itemize}
\end{proof}

\end{document}